%% file: paper.tex
\documentclass[%
reprint,
nofootinbib,
mathtools,
amssymb,
aps,
prx,
superscriptaddress,
usenames,dvipsnames,
floatfix
]{revtex4-2}

\input{paper_preamble}

\begin{document}

\title{Symmetry-Accelerated Classical Simulation of Clifford-Dominated Circuits}

\author{Giulio Camillo}
\thanks{Authors GC and FCRP contributed equally.}
\affiliation{\ugrshort}

\author{Filipa C. R. Peres}
\email[Corresponding author: ]{fcrperes@onsager.ugr.es}
\affiliation{\ugrshort}

\author{Markus Heinrich}
\email[Corresponding author: ]{markus.heinrich@uni-koeln.de}
\affiliation{Institute for Theoretical Physics, University of Cologne, Cologne, Germany}

\author{Juani Bermejo-Vega}
\affiliation{\ugrshort}
\affiliation{Institute Carlos I for Theoretical and Computational Physics, University of Granada, Granada, Spain.}

\date{\today}

\begin{abstract}
Classical simulation of quantum circuits plays a crucial role in validating quantum hardware and delineating the boundaries of quantum advantage. Among the most effective simulation techniques are those based on the stabilizer extent, which quantifies the overhead of representing non-Clifford operations as linear combinations of Clifford unitaries. However, finding optimal decompositions rapidly becomes intractable as it constitutes a superexponentially large optimization problem. In this work, we exploit symmetries in the computation of the stabilizer extent, proving that for real, diagonal, and real-diagonal unitaries, the optimization can be restricted to the corresponding subgroups of the Clifford group without loss of optimality. This ``strong symmetry reduction'' drastically reduces computational cost, enabling optimal decompositions of unitaries on up to seven qubits using a standard laptop---far beyond previous two-qubit limits.  Additionally, we employ a ``weak symmetry reduction'' method that leverages additional invariances to shrink the search space further. Applying these results, we demonstrate exponential runtime improvements in classical simulations of quantum Fourier transform circuits and measurement-based quantum computations on the Union Jack lattice, as well as new insights into the nonstabilizer properties of multicontrolled phase gates and unitaries generating hypergraph states. Our findings establish symmetry exploitation as a powerful route to scale classical simulation techniques and deepen the resource-theoretic understanding of quantum advantage.
\end{abstract}

\maketitle

\section{Introduction}\label{sec: Introduction}

The central promise of quantum computing is the solution of problems that are out of reach for classical machines.
While we thus expect quantum computers to be generally hard to simulate by fully classical means, research into classical simulation algorithms and techniques has proven crucial in developing a better understanding of the precise onset of such a quantum advantage.
Indeed, classical simulators allow direct validation of quantum hardware and serve as a benchmark to compare quantum advantage claims against \cite{arute_quantum_2019,kim_evidence_2023,morvan_phase_2024}, while still regularly defying them
\cite{begusic_fast_2024,tindall_efficient_2024}.
Improved simulation techniques thus also raise the bar for better and bigger quantum prototypes that are and will be arising. 
Moreover, classical simulation provides valuable insights into which quantum resources and phenomena enable the desired quantum advantage. In turn, this can inform the development of more efficient quantum algorithms. 

The prototypical example of such a deeper understanding provided by classical simulation is the Gottesman-Knill theorem:
It states that even highly entangling quantum circuits---so-called \emph{stabilizer circuits}---are efficiently classically simulable as long as they only involve Clifford gates as well as preparation and measurement in the computational basis~\cite{PhDGottesman}.
This shows that, besides entanglement, \emph{nonstabilizerness} (also known as \emph{magic}) is a necessary resource for quantum advantage. 
Strikingly, this resource can be provided by a single non-Clifford gate~\cite{Nebe2001}---often chosen as the $T$ gate---or, alternatively, by a continuous supply of \emph{magic states} \cite{Bravyi2005}.

The simulation of stabilizer circuits can be extended to arbitrary quantum circuits at the cost of an exponential overhead in the number of nonstabilizer elements (e.g.,~magic states or non-Clifford gates) and many works have successively developed more and more efficient algorithms to this end \cite{Aaronson2004,BravyiSmolin2016, BravyiGosset2016, Howard2017, Bravyi2019, Seddon2021}.
These approaches typically decompose a quantum circuit into a linear combination of stabilizer circuits; doing so on the level of the whole circuit is, however, generally computationally intractable.
Indeed, previous works have instead been limited to the decomposition of single or two-qubit gates.
This makes the runtime of these approaches very sensitive to compilation efficiency, in particular, the ability to minimize nonstabilizer elements in the circuit, but also blind to known submultiplicative effects in nonstabilizer resources:
The nonstabilizerness of a circuit is almost always strictly less than the product of its parts.
Taken together, these issues can severely limit the applicability of stabilizer-based simulation algorithms.

In this work, we address this problem by exploiting the concept of \emph{symmetry}, providing methods that enable us to exactly and optimally decompose entire nonstabilizer subcircuits (on up to seven qubits).
This allows us to show that submultiplicative effects in stabilizer-based simulation are, in fact, drastic: Our novel methods lead to savings of several orders of magnitude in runtime for paradigmatic circuits.

Concretely, we consider simulation algorithms based on the \emph{stabilizer extent}, which rank among the most successful and best studied approaches \cite{Bravyi2019, Seddon2021, Qassim2019, pashayan_fast_2022, qassim_improved_2021}.
By exploiting the symmetry of highly relevant classes of non-Clifford unitaries, we achieve significant reductions in the size of the underlying optimization problem. 
More specifically, we prove that if a unitary $U$ is real and/or diagonal, we can find an optimal decomposition of $U$ in terms of real and/or diagonal Clifford unitaries only, thus avoiding full optimization over the much larger Clifford group.
Strikingly, such a `strong' reduction is not possible for every symmetry.
We can, however, exploit any other symmetries of $U$ to perform a `weak' reduction that still leads to a further and strong decrease of the search space.

Importantly, such symmetries are often present in multiqubit gates or subcircuits that appear as building blocks of quantum algorithms.
Examples of this are multicontrolled rotation and phase gates, Toffoli gates, Givens rotations, or the Grover diffusion operator.

We use these findings to accelerate classical simulation based on the sum-over-Cliffords method~\cite{Bravyi2019}. By grouping multiple non-Clifford gates and jointly decomposing them, we demonstrate large savings over the simulation of individual gates due to submultiplicative effects~\cite{Qassim2019}.
Using our symmetry reduction techniques, we can compute such decompositions for symmetric unitaries supported on at most seven qubits using a consumer laptop, significantly exceeding the previous state of the art, limited to two qubits \cite{Qassim2019}. We explicitly demonstrate these improvements in the simulation of quantum Fourier transform (QFT) circuits and of (arbitrary, universal) measurement-based quantum computations (MBQCs). To list some examples, a 16-qubit  QFT circuit can be simulated two orders of magnitude faster using our results than using prior state-of-the-art stabilizer-based techniques. Additionally, our results enable an exponential improvement of the classical (weak) simulation of any (universally general) MBQC performed on an $n$-qubit Union Jack lattice and using only Pauli measurements. From a resource-theoretical perspective, our work allows the in-depth study of the nonstabilizer properties of important families of circuits using the stabilizer extent. Along this line, we make an extensive analysis of the properties of multicontrolled phase gates $C^{n-1}P (\theta)$ as well as of the unitaries generating hypergraph and generalized hypergraph quantum states, thoroughly investigating submultiplicativity effects and also the role played by entanglement in the optimal decompositions.

Finally, let us outline the structure of this paper.
Sec.~\ref{sec: Background} gives a brief introduction to the relevant background. 
In Sec.~\ref{sec: Results}, we present our theoretical results on the strong and weak symmetry reduction of the stabilizer extent and provide examples for when strong reduction is possible and when it is not. Sec.~\ref{sec: Practical demonstrations} is concerned with practical demonstrations of the benefits of our theoretical results:
We analyze the nonstabilizerness of multicontrolled phase gates and show how our techniques lead to exponential improvements in the sum-over-Cliffords simulation of QFT circuits, as well as in MBQC using hypergraph states. In Sec.~\ref{sec: Conclusions}, we conclude with a discussion of our most important results and give an outlook on the impact of this paper and prospective future research.

\section{Background}\label{sec: Background}

\subsection{Classical simulation and the stabilizer extent}\label{subsec: Background - stabilizer extent}

Stabilizer-based classical simulation algorithms stem from the fact that generic circuits can always be expressed over a Clifford+$K$ gate set, where $K$ stands for a single or multiple non-Clifford gate(s).
These non-Clifford unitaries are then simulated using stabilizer operations only.
Since the latter are by themselves efficiently simulable by Gottesman-Knill, the runtime overhead comes entirely from the simulation of the non-Clifford circuit components.

Most of the non-Clifford gates that are typically considered, for instance, the $T$, $CCZ$, and Toffoli gates, can be implemented using a circuit gadget that consumes a resource state--called \emph{magic state} in this context.
Since the gadget itself is composed of stabilizer operations, the difficulty in simulating those non-Clifford gates is then given by the simulation of the magic states.
To this end, we can express the magic state as a linear combination of stabilizer states $\ket{s}$, $\ket\psi = \sum_s x_s \ket{s}$, and use the linearity of the final state's amplitudes to reduce their computation to a series of stabilizer circuit simulations and suitable pre- and post-processing~\cite{BravyiGosset2016,Bravyi2019}.

An alternative, more general approach is given by the sum-over-Cliffords simulation~\cite{Bravyi2019}, wherein the non-Clifford unitaries are directly decomposed into a linear combination of Clifford unitaries $C$, say $U= \sum_C x_C C$. 
In perfect analogy with the gadget-based approach, we can then use the Gottesman-Knill theorem to simulate the stabilizer circuits resulting from the decomposition and suitably post-process the results. 
This method has several advantages over gadget-based simulation: 
First, it is applicable to arbitrary non-Clifford unitaries and not only those enabled by magic-state injection; secondly, it is simpler to implement; and, finally, it often involves the manipulation of quantum states with fewer qubits than those in the gadget-based approach, making it faster~\cite{Bravyi2019}.

In both gadget-based and sum-over-Cliffords simulation, the simulation runtime is proportional to the number of terms in the (non-unique) decomposition of the target states or unitaries.
Thus, the simulation works best with a minimal number of terms, also known as the \emph{stabilizer rank} of the target \cite{BravyiSmolin2016,BravyiGosset2016}.
The stabilizer rank can, in principle, be computed by minimizing the so-called $\ell_0$-norm of the coefficient vector $\mathbf{x}$, a problem that is, however, known to be $\mathsf{NP}$-complete~\cite[Appendix~A6]{Garey1979computers}. 
In practice, even an approximate low-rank decomposition would be sufficient. 
\textcite{Bravyi2019} show how to probabilistically construct a $\delta$-approximate decomposition with rank $\|\mathbf{x}\|_1^2/\delta^2$ from an exact decomposition with coefficient vector $\mathbf{x}$. 
Hence, finding decompositions with small $\ell_1$-norm is central to this approach.
Importantly, this problem can now be phrased as conic programming and can be solved using standard algorithms.
This motivates the definition of the \emph{stabilizer extent} for pure quantum states and matrices as follows.

\begin{definition}[Stabilizer extent for pure states~\cite{Bravyi2019}]
    Let $\mathrm{STAB}_n$ denote the set of $n$-qubit pure stabilizer states. The \emph{stabilizer extent} of an arbitrary, $n$-qubit quantum state $\ket{\psi}$ is defined as:
    \begin{align*}
        \xi (\ket{\psi}) \coloneqq \min \{ & \lVert \mathbf{x} \rVert_1^2 \; \vert \; \mathbf{x} \in \mathbb{C}^{\left|\mathrm{STAB}_n\right|} \,:\\
        & \ket{\psi} = \sum_{\ket{s} \in \mathrm{STAB}_n} x_{s} \ket{s} \}\,.
    \end{align*} 
We call a decomposition of $\ket\psi$ into stabilizer states \emph{optimal} if it is the minimizer of the above optimization problem.
\end{definition}

\begin{definition}[Stabilizer extent for matrices~\cite{Bravyi2019}]\label{def: Stabilizer extent for unitaries}
    Let $\mathcal{C}_n$ be the $n$-qubit Clifford group. We define the \emph{stabilizer extent} of a matrix $M\in \mathbb{C}^{2^n \times 2^n}$ as:
    \begin{equation*}
        \xi (M) \coloneqq \min \{ \lVert \mathbf{x} \rVert_1^2 \; \vert \; \mathbf{x} \in \mathbb{C}^{\left|\Cl{n}\right|} \,: \, M = \sum_{C\in \Cl{n}} x_CC\}\,.
    \end{equation*}
    We again call a decomposition of $M$ optimal if it has $\lVert \mathbf{x} \rVert_1^2 = \xi$.
\end{definition}

The stabilizer extent has many useful properties that often allow for simplifying or avoiding its computation and finding (sub)-optimal decompositions from known ones.
A powerful property of the state version is its multiplicativity under tensor products of few-qubit quantum states \cite{Bravyi2019}:
\begin{equation}
\label{eq:stabilizer-extent-multiplicativity}
    \xi (\bigotimes_{j=1}^m \ket{\psi_j}) = \prod_{j=1}^m \xi (\ket{\psi_j}) \,,
\end{equation}
where each $\ket{\psi_j}$ is supported on at most three qubits. 
This result significantly facilitates the runtime estimation of classical simulation in the magic state model, since typically considered magic states are indeed such few-qubit states.
Assuming, for instance, that the circuit of interest contains $k$ $T$ gates, the runtime cost of simulating it through gadgetization is given by $\xi(\ket{T}^{\otimes k}) = \xi(\ket T)^k = \cos(\pi/8)^{-2k}$ where $\ket{T}\coloneqq T\ket{+} = (\ket{0} + e^{i\pi/4}\ket{1})/\sqrt{2}$~\cite{Bravyi2019}.
We note that the stabilizer extent is generally submultiplicative for states on more than three qubits~\cite{Heimendahl2021}, thus requiring the direct solution of the defining optimization problem.

Notably, it is not known whether the matrix version of the stabilizer extent is multiplicative under tensor products of few-qubit unitaries.
In some cases, we can exploit the following relation between the state and matrix version, which follows directly from the definitions:
\begin{equation}\label{eq: relation extent state <-> unitary}
    \xi (U\ket{s}) \leq \xi (U)\,.
\end{equation}
Here, $\ket{s}$ is an arbitrary stabilizer state.
An important case is given by a diagonal $U$ and the state $\ket{U}=U\ket{+}^{\otimes n}$.
Although these are essentially the same objects, equality in Eq.~\eqref{eq: relation extent state <-> unitary} is only known in certain cases, for instance, if $U$ is a single-qubit $Z$-rotation, $R_Z (\theta) = e^{-i\frac{\theta}{2} Z}$ (of which the $T$ gate is a special case) or the $CCZ$ gate.
In these cases, we can exploit the multiplicativity of the stabilizer extent for states to conclude that $\xi (T^{\otimes k}) = \xi (T)^k = \cos (\pi / 8)^{-2k}$ and $\xi (CCZ^{\otimes k}) = \xi (CCZ)^k = (16/9)^k$ \cite{Bravyi2019}.

Despite the possible lack of multiplicativity, the stabilizer extent for matrices has several other important properties which we summarize in the following proposition and prove in Appendix~\ref{app: Properties extent} for completeness.
\begin{proposition}\label{prop: Properties extent}
    The stabilizer extent has the following properties:
    \begin{enumerate}[label=(\roman*)]
        \item  $\xi(A)=1$ for any Clifford unitary $A\in \mathcal{C}_n$.
        
        \item Invariance under multiplication by Clifford unitaries: $\xi(A) = \xi(CA) = \xi(AC^{\prime})\quad \forall\, C,C^{\prime}\in \mathcal{C}_n$.\\
        Monotonicity under stabilizer code projectors $\Pi$: $\xi(\Pi A) \leq \xi(A)$, $\xi(A\Pi) \leq \xi(A)$.
        
        \item Submultiplicativity under matrix multiplication and tensor products: 
        $\xi(AB) \leq \xi(A)\xi(B)$, $\xi(A\otimes B) \leq \xi(A)\xi(B)$.
        Moreover: $\xi(A\otimes B) = \xi(B\otimes A)$.
        
        \item Convexity: $\xi(t A + (1-t)B) \leq t \xi(A) + (1-t)\xi(B)$ for all $t\in[0,1]$.\\
        Homogeneity: $\xi(xA) = |x|^2\xi(A)$ for $x\in\C$.
    \end{enumerate}
\end{proposition}

We note that the stabilizer extent also provides a faithful measure of nonstabilizerness or magic:
That is, $\xi(\ket\psi)=1$ if and only if $\ket\psi$ is a stabilizer state, and otherwise $\xi(\ket\psi)>1$.
The same holds for the matrix version if we restrict to unitaries:
$\xi(U) = 1$ iff $U\in \Cl{n}$ and $\xi(U)>1$ otherwise. 
For non-unitary matrices, however, $\xi$ can assume values smaller than 1.
The monotonicity then makes the stabilizer extent a \emph{magic monotone}, which can be used for lower bounding the cost of (exact) synthesis of certain unitaries~\cite{Howard2017, Beverland2020lower}.
Indeed, combining the properties of the stabilizer extent in Proposition~\ref{prop: Properties extent} leads to the following statement.

\begin{proposition}[Exact gate synthesis]\label{prop: gate synthesis}
    For any unitary $U\in \mathrm{U}(2^n)$, there always exists an $s$ so that $ \xi (T)^{s-1} < \xi (U) \leq \xi (T)^{s}\,.$ If $U$ is \emph{exactly synthesizable} in terms of the Clifford+$T$ gate set, then $U$ requires at least $s$ $T$ gates to be synthesized.
\end{proposition}

We use this proposition to discuss some of our results in Sec.~\ref{sec: Practical demonstrations}. However, since this is by no means central to our work, we leave the proof to Appendix~\ref{app: Synthesis}.

Given a quantum circuit over a Clifford+$K$ gate set, the above-discussed properties also allow us to construct the necessary decompositions for a classical simulation algorithm, the runtime of which can, however, depend highly on the precise construction.
One crucial reason for this is the submultiplicativity under matrix multiplication, cf.~Proposition \ref{prop: Properties extent} (iii).
As we will also see later in Sec.~\ref{sec: Practical demonstrations}, it can be highly beneficial to `block' non-Clifford gates and try to decompose their matrix product instead of decomposing them one-by-one.
For this approach, we can typically only rely on the direct computation of the stabilizer extent by solving the defining optimization problem.
Although the latter is a second-order cone program (SOCP) which can be solved in polynomial time in the problem size, the superexponentially growing size of $\mathrm{STAB}_n$ and $\Cl{n}$ with $n$ prevents us from doing so in practice. 
Even small instances require supercomputing resources, cf.~Table~\ref{tab: Space and time comparisons different problems} in Appendix~\ref{app: practical_considerations}.
Moreover, the matrix version is generically harder than the state version due to the square increase in the dimension of the underlying vector space and the fact that the Clifford group is superexponentially larger than the size of the set of stabilizer states:
\begin{equation*}
    \frac{\left| \Cl{n} \right|}{\left| \mathrm{STAB}_n \right|} = \frac{2^{n^2 + 2n} \prod_{j=1}^n (4^j - 1)}{2^n \prod_{j=1}^n (2^j + 1)} = \mathcal{O} \left( 2^{n^2}\right)\,.
\end{equation*}

While impressive progress has been made to speed up the computation of the stabilizer extent for states up to 10 qubits, e.g., by exploiting linear dependencies between stabilizer states \cite{deSilva2024} or suitably restricting the set of stabilizer states used for optimization \cite{Hamaguchi2025}, the matrix version has received little attention up to now.
In this work, we address this gap in the literature and study how symmetries can be exploited to accelerate the computation of the stabilizer extent for unitaries in many important cases.

\subsection{Symmetry and invariance}

Central to our work are the notions of \emph{symmetry} and \emph{invariance}. 
Given a matrix $M$, we say that it is invariant under a map $\phi$ if $\phi(M)=M$.
In the following, we focus on invertible maps that are linear or anti-linear and preserve the Hilbert-Schmidt inner product, that is, $\tr(\phi(A)^\dagger \phi(B)) = \tr(A^\dagger B)$. 
We call such maps \emph{isometries}.
The set of isometries that leave $M$ invariant forms a group, henceforth called the \emph{symmetry group} of $M$.
Given such a group of isometries $\G$, the set of invariant matrices forms a subspace, the so-called $\G$-\emph{invariant subspace}.
Note that if $\G$ contains anti-linear isometries, then this is only a real subspace.
In the following, we shall focus on finite groups, although the generalization to compact groups is evident.
It is an elementary fact that the projection onto the $\G$-invariant subspace is given by
\begin{equation*}
    \Pi_{\G} (M) := \frac{1}{|\G|} \sum_{\phi\in\G} \phi(M)\,.
\end{equation*}
Again, if $\G$ contains anti-linear isometries, then $\Pi_\G$ is only $\R$-linear, not $\C$-linear.

The main symmetry groups we consider in this work are the following:\\[0.5em]
\noindent\textbf{Complex conjugation:} Let $\mathcal{K}_n = \{\id, (\cdot)^\ast\}$ be the group generated by complex conjugation $M\mapsto M^*$ in the computational basis.
Its invariant subspace is the real subspace of matrices with real entries, $\R^{2^n\times 2^n}$.\\[0.5em]
\noindent\textbf{Transposition:} Let $\mathcal{T}_n = \{\id, (\cdot)^T\}$ be the group generated by transposition $M\mapsto M^T$ in the computational basis. 
Its invariant subspace is the complex subspace of symmetric matrices, $\Sym{2^n}(\C)$.\\[0.5em]
\noindent\textbf{Diagonal Pauli group:} Let $\ZP{n} = \{ Z(z):=Z^{z_1}\otimes\dots\otimes Z^{z_n} \; | \; z \in \Z_2^n \}$ be the group of diagonal Pauli operators, i.e., the one generated by Pauli $Z$ matrices on every qubit.
We can identify $\ZP{n}$ with the group of quantum channels $M\mapsto Z(z)M Z(z)^\dagger$, the invariant subspace of which is the complex subspace of diagonal matrices.\\[0.5em] 
\noindent\textbf{Permutation group:} Let $\Sn{n}$ be the group of qubit permutations generated by all possible SWAP gates.
Again, by a slight abuse of notation, we consider the invariant subspace under conjugation with permutations, $M\mapsto \pi M \pi^\dagger$.

\vspace{0.5em}

We are furthermore interested in the subset of Clifford unitaries $\Cl{n}^\G$ that are invariant under a group $\G$.
If $\G$ acts by conjugation, such as for complex conjugation, the diagonal Pauli group, and the permutation group, then $\Cl{n}^\G$ is a subgroup of the Clifford group.
In contrast, $\Cl{n}^{\mathcal{T}_n}$ is not a subgroup as transposition-invariant Cliffords do not necessarily commute.\footnote{For instance, Hadamard and $CZ$ are symmetric matrices, but $(H_1 CZ)^T = CZ^T H_1^T = CZ H_1 \neq H_1 CZ$.}
The subgroups $\Cl{n}^{\mathcal{K}_n}$ and $\Cl{n}^{\ZP{n}}$ are known as the \emph{real Clifford group} $\Rl{n}$ and \emph{diagonal Clifford group} $\Di{n}$, respectively.
The former is generated by the Pauli $Z$, CNOT, and Hadamard gates, while the latter is generated by the phase gate $S$ and the controlled-$Z$ gate $CZ$.
We will also consider the group $\K{n}\times\ZP{n}$ generated by the commuting groups $\K{n}$ and $\ZP{n}$.
The invariant Clifford subgroup is then the group of real-diagonal Cliffords $\RD{n}$, generated by $Z$ and $CZ$.

\section{Results}\label{sec: Results}

\subsection{Computation of stabilizer extent under symmetries}\label{subsec: Statement of problem and main theoretical results}

In this section, we show that the computation of the stabilizer extent for unitaries that are invariant under certain symmetry groups $\G$ can be reduced to an optimization problem over $\Cl{n}^\G$ only, which is typically of a much smaller size than the original one.
To this end, we define the $\G$-symmetric stabilizer extent as follows.

\begin{definition}[$\G$-symmetric stabilizer extent]\label{def: Symmetric stabilizer extent}
    Take $\G$ a symmetry group and suppose that $\Cl{n}^\G$ spans the (real or complex) $\G$-invariant subspace. Then, the \emph{$\G$-symmetric stabilizer extent} of a $\G$-invariant matrix $M\in \mathbb{C}^{2^n\times 2^n}$ is defined as:
    \begin{equation*}
        \xi_{\G} (M) \coloneqq \min \{ \lVert \mathbf{x} \rVert_1^2 \; \vert \; \mathbf{x} \in \mathbb{C}^{\left|\Cl{n}^\G\right|} \;: \; M = \sum_{C\in\Cl{n}^\G} x_CC\}\,.
    \end{equation*}
\end{definition}

In this work, we focus on the cases where $\G = \K{n}, \ZP{n}, \K{n}\times\ZP{n}$.
Then, $\Cl{n}^\G$ corresponds to the real, diagonal, and real-diagonal subgroups of the Clifford group, respectively.
Since already the respective Pauli subgroups span the $\G$-invariant subspaces in these cases, $\xi_G$ is well-defined.
In the $\ZP{n}$ and $\K{n}\times\ZP{n}$ case, the matrix $U$ and the Cliffords $C\in\Cl{n}^\G$ are real, and we can thus restrict the expansion coefficients $\vec x$ to real numbers, turning the optimization problem into a linear program (LP).

For $\G = \K{n}, \ZP{n}, \K{n}\times\ZP{n}$, we then show that the stabilizer extent of a $\G$-invariant unitary coincides with its $\G$-symmetric stabilizer extent.

\begin{theorem}[Strong $\G$-symmetry reduction]\label{theorem: Main result -- diagonals and reals}
    Let $\G\in \{\K{n}, \ZP{n}, \K{n}\times\ZP{n}\}$. For any $U$ such that $U = \Pi_{\G}(U)$, $\xi(U) = \xi_{\G} (U)$.
\end{theorem}

The proof of Theorem \ref{theorem: Main result -- diagonals and reals} relies on the following intermediate result, the proof of which we defer to Appendix~\ref{app: Proof of characterization}.

\begin{lemma}\label{lemma: Characterization}
    Let $\G\in \{\K{n}, \ZP{n}, \K{n}\times\ZP{n}\}$, $C\in \Cl{n}$, and $x\in\C$.
    Then, $\xi_{\G} (\Pi_\G (x C)) \leq |x|^2$.
\end{lemma}

\begin{proof}[Proof of Theorem~\ref{theorem: Main result -- diagonals and reals}]
    Suppose $U$ has an optimal decomposition $U=\sum_{C\in \Cl{n}} x_C C$, i.e.~$\xi (U) = \lVert \mathbf{x} \rVert_1^2\,.$ We then have $U = \Pi_{\G} (U) = \sum_{C \in \Cl{n}} \Pi_{\G} (x_C C)$, as $\Pi_{\G}$ is additive (but not necessarily $\C$-linear).
    Next, let $\Pi_\G (x_C C) = \sum_{S\in\Cl{n}^\G} y_S^{(C)} S$ be an optimal decomposition, and write 
    \begin{equation*}
        U 
        = \sum_{C \in \Cl{n}} \sum_{S\in \Cl{n}^\G} y_S^{(C)} S
        =
        \sum_{S\in \Cl{n}^\G}
        \left(\sum_{C \in \Cl{n}}y_S^{(C)} \right) S \,.
    \end{equation*}
    By Lemma~\ref{lemma: Characterization}, $\lVert \mathbf{y}^{(C)} \rVert_1 \leq |x_C|$, thus
    \begin{multline}
    \sqrt{\xi_\G(U)} \leq 
        \sum_{S\in \Cl{n}^\G} \left| \sum_{C \in \Cl{n}} y_S^{(C)} \right|
        \leq
        \sum_{C \in \Cl{n}} \lVert \mathbf{y}^{(C)} \rVert_1 \\
        \leq \lVert \mathbf{x} \rVert_1
        = \sqrt{\xi(U)} \,. \nonumber
    \end{multline}
    Since $\xi (U) \leq \xi_{\G} (U)$ by definition, we have shown equality.
\end{proof}
\begin{figure}
    \centering
    \includegraphics[width=0.99\linewidth]{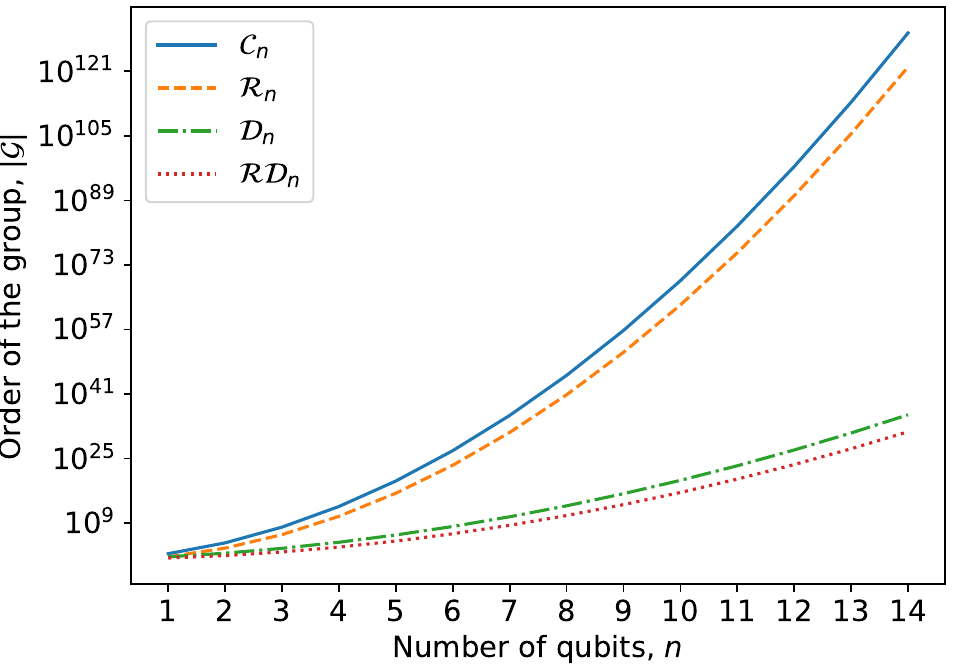}
    \caption{\textbf{Order of the Clifford group $\Cl{n}$ and subgroups of interest.} The figure presents the size of the Clifford group $\Cl{n}$ in solid blue, the real subgroup $\Rl{n}$ in dashed orange, the diagonal subgroup $\Di{n}$ in dash-dotted green, and the real-diagonal subgroup in dotted red.}
    \label{fig: Cardinality}
\end{figure}

In Fig.~\ref{fig: Cardinality}, we see that, by itself, Theorem~\ref{theorem: Main result -- diagonals and reals} allows a reduction of orders of magnitude in the number of elements required for the computation of the stabilizer extent, in particular for the cases of diagonal or real-diagonal unitaries. 
However, the invariant subgroups are still superexponentially large and we find, for instance, $\left| \Di{n} \right| > 10^8$ for $n=6$, which is already too demanding for the optimizer. Fortunately, this result can be enhanced by leveraging other underlying symmetries of the nonstabilizer unitary of interest. 

As an example, let us consider a diagonal unitary $U$. By Theorem~\ref{theorem: Main result -- diagonals and reals}, we know it admits an optimal decomposition so that $U = \sum_{D\in \Di{n}} x_D D$.
Let us, in addition, assume that $U$ is also invariant under another group of \emph{linear} isometries $\G$.
While Theorem~\ref{theorem: Main result -- diagonals and reals} may not apply to this additional symmetry (cf.~Sec.~\ref{subsec: Counter examples to T1}), we can still exploit the linearity to write
\begin{equation*}
    \Pi_{\G} (U) = U = \sum_{D\in\Di{n}} x_D \Pi_{\G} (D) \,.
\end{equation*}
Hence, the optimization can be carried out over the set of projections $\Pi_{\G} (D)$.
Coincidences among those can significantly reduce the required search space and thus the size of the optimization problem. An analogous strategy was previously used for other magic monotones \cite{Heinrich2019}.
Moreover, in the case of many important symmetry groups, each $\Pi_{\G}(D)$ is a convex combination of Clifford unitaries, allowing us to find not only the optimal value of the stabilizer extent but also a suitable decomposition in terms of Clifford unitaries. We formalize these ideas into the following lemma.
\begin{lemma}[Weak $\G$-symmetry reduction]\label{lemma: weak reduction} Let $U$ be an $n$-qubit unitary whose optimal decomposition can be found over a generic set of matrices, $\mathcal{S}$, so that $U = \sum_{C\in \mathcal{S}} x_C C\,.$ Let $\G$ be a symmetry group of $U$: $\Pi_{\G}(U) = U$. Then, the stabilizer extent of $U$ can be computed using only the set of different projections of the matrices in the set $\mathcal{S}$.
\end{lemma}

Unlike Theorem~\ref{theorem: Main result -- diagonals and reals}, this lemma is completely general. The gain associated with its application is upper-bounded by the number of orbits resulting from the projection of the elements of the original set of matrices. In Appendix~\ref{app: cardinality of Dn/Sn}, we illustrate the computation of this upper bound for diagonal unitaries that are also permutation-invariant, i.e., when $\G=\mathcal{S}_n$. In Appendix~\ref{app: Permutation invariance}, we go one step further and fully enumerate the total number of nonequivalent projections of diagonal Cliffords onto the permutation invariant subspace. This leads to impressive reductions in the size of the search space for the SOCP. As an example, for $n=6$, we reduce the number of terms from $\sim 1.3\times 10^{8}$ to only $\sim 6.4\times 10^{5}$, a mere 0.5\% of the original number.

\begin{figure*}[t]
    \centering
    \includegraphics[width=0.95\linewidth]{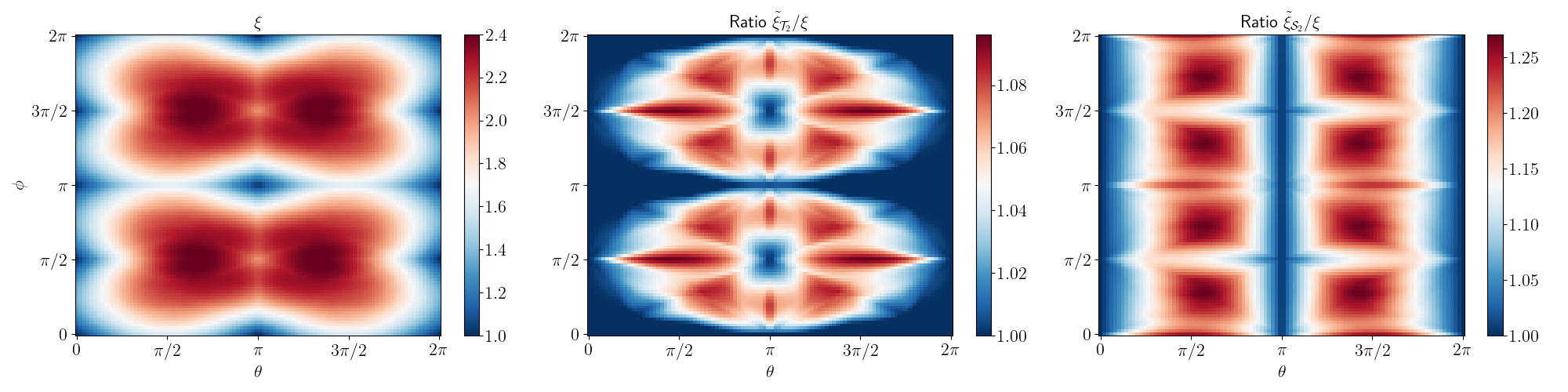}
    \caption{\textbf{Stabilizer extent of two-qubit $\text{fSim}(\theta,\phi)$ gates as a function of the parameters $\theta$ and $\phi$.} Color map depicting (a) the (true) stabilizer extent, $\xi \left(\text{fSim}(\theta,\phi) \right)$, (b) the ratio between the minimum $\ell_1$-norm squared obtained only with the transposition-invariant Cliffords and the (true) stabilizer extent, $\tilde{\xi}_{\mathcal{T}_{2}}\left(\text{fSim}(\theta,\phi) \right) /\xi\left(\text{fSim}(\theta,\phi) \right)$, and (c) the ratio between the minimum $\ell_1$-norm squared obtained using the permutation-invariant Cliffords and the (true) stabilizer extent, $\tilde{\xi}_{\mathcal{S}_{2}}\left(\text{fSim}(\theta,\phi) \right) /\xi\left(\text{fSim}(\theta,\phi) \right)\,.$ If the result in Theorem~\ref{theorem: Main result -- diagonals and reals} were to hold for this particular example with the transposition and permutation invariances, the ratios depicted in (b) and (c) would have to be 1 for every pair $(\theta,\phi)$. That is eminently not the case, even though it is possible to expand $\text{fsim}(\theta,\phi)$ in terms of $\Cl{2}^{\mathcal{T}_2}$ and $\Cl{2}^{\Sn{2}}$.}
    \label{fig: xi fSim gates}
\end{figure*}

In our practical demonstrations, we focus on leveraging Theorem~\ref{theorem: Main result -- diagonals and reals} for the case of diagonal and real-diagonal unitaries, since these lead to the strongest reductions as seen by Fig.~\ref{fig: Cardinality}. Additionally, as we will discuss, many important circuits use gates of this form. To conclude, we would like to take this opportunity to point out that our Theorem~\ref{theorem: Main result -- diagonals and reals} establishes an important connection between the stabilizer extent of quantum states and the stabilizer extent of diagonal unitaries. Let $U$ be a diagonal unitary. Our result guarantees that the optimal decomposition of $U$ (with respect to the $\ell_1$-norm) can be obtained by decomposing $U$ only in terms of diagonal Clifford unitaries $D$ so that $U = \sum_{D\in \Di{n}} x_D D\,.$ Let $\ket{U} = U \ket{+}^{\otimes n}$. Then, the quantum state $\ket{U}$ admits a decomposition of the form: $\ket{U} = \sum_{D\in \Di{n}} x_D \ket{s_D}$, where $\ket{s_D} = D \ket{+}^{\otimes n}$. Since there is a one-to-one correspondence between the diagonal gates and the set of equatorial stabilizer states, this means that the problem of finding the optimal extent of $U$ can be reduced to the problem of finding the optimal $\ell_1$-norm of a decomposition of $\ket{U}$ using only equatorial stabilizer states. Hence, if the methods proposed in Refs.~\cite{deSilva2024, Hamaguchi2025} can be modified to find the optimal decomposition of magic states exclusively in terms of equatorial stabilizer states, they might provide a nice way of boosting the results demonstrated below.

\subsection{Is strong symmetry reduction extendable to every symmetry group?}\label{subsec: Counter examples to T1}

It could be tempting to assume that Theorem~\ref{theorem: Main result -- diagonals and reals} is generalizable to include every possible symmetry group. 
However, for the $\G$-symmetric stabilizer extent to be well-defined, the subset $\Cl{n}^{\G}$ must span the entire $\G$-invariant subspace, otherwise the expansion $M=\sum_{C\in \Cl{n}^{\G}} x_C C$ is not guaranteed to exist (cf.~Def.~\ref{def: Symmetric stabilizer extent}). 
The three symmetry groups $\G = \K{n}, \ZP{n}, \K{n}\times\ZP{n}$ for which we proved strong $\G$-symmetry reduction all meet this condition. For other $\G$, however, this might not be the case and thus our result cannot be applied (see proofs in Appendix~\ref{app: Proof of characterization}). 
One example is $\Cl{n}^{\Sn{n}}$, which does not span the permutation-invariant subspace for $n>3$.\footnote{The Clifford group is not a 4-design, thus the set of operators commuting with all $C^{\otimes 4}$ is larger than the span of $\Sn{4}$. 
If $\Cl{n}^{\Sn{n}}$ span the entire permutation-invariant subspace, there would be operators in the commutant of the unitary group that are not linear combinations of permutations, contradicting Schur-Weyl duality.}

What if $\Cl{n}^{\G}$ does not span the entire $\G$-invariant subspace, but a specific gate $U$ so that $U = \Pi_\G(U)$ \emph{can} be expanded in terms of that set? 
In this restricted setting, could a similar result as the one provided by Theorem~\ref{theorem: Main result -- diagonals and reals} apply? 
That is, could we state that $\tilde{\xi}_\G (U) = \xi (U)$ (where $\tilde\xi_\G$ is now restricted to the span of $\Cl{n}^{\G}$)?
We answer this question in the negative by providing a specific counterexample. 
In particular, we consider the important set of two-parameter fermionic simulation (fSim) gates~\cite{Kivlichan+2018}, which have the following matrix representation in the computational basis:
\begin{equation}\label{eq: fSim gates}
    \text{fSim} (\theta,\phi) \coloneqq \begin{pmatrix}
       1 & 0 & 0 & 0 \\
       0 & \cos (\theta) & -i\sin (\theta) & 0 \\
       0 & -i\sin (\theta) & \cos (\theta) & 0 \\
        & 0 & 0 & e^{-i\phi}
    \end{pmatrix}\,.
\end{equation}
It is easily understood from the above representation that these gates enable continuous transformations, with $\theta$ controlling the swapping between $\ket{01}$ and $\ket{10}$ and $\phi$ introducing a phase for $\ket{11}$. Important particular examples of these gates are $CP(\phi) = \text{fSim}(0,-\phi)$ and $\text{iSWAP} = \text{fSim}(-\pi/2,0)$; $\text{fSim}(\theta,0)$ is a matchgate~\cite{Valiant2002} and $\text{fSim}(\pi/2,\pi/6)$ was the two-qubit gate implemented in Google's quantum processor in their attempted quantum supremacy experiment~\cite{GoogleAI2019}. Since these are only two-qubit gates, it is feasible to compute the extent using the entire Clifford group. The results are presented in Fig.~\ref{fig: xi fSim gates}(a).

Looking at Eq.~\eqref{eq: fSim gates}, we can see that the fSim gates are invariant under transposition: $\text{fSim}(\theta,\phi) = \text{fSim}(\theta,\phi)^{T}$.
We can thus attempt to compute the optimal decomposition of $\text{fSim}(\theta,\phi)$ into symmetric Clifford gates $\Cl{2}^{\T{2}}$ and compare the associated squared $\ell_1$-norm $\tilde{\xi}_{\T{2}}(\text{fSim}(\theta,\phi))$ to the stabilizer extent $\xi(\text{fSim}(\theta,\phi))$.
Our numerical studies demonstrate that, while for some parameter ranges $\tilde{\xi}_{\mathcal{T}{2}} (\text{fSim}(\theta,\phi))$ corresponds to the true stabilizer extent, for the most part $\tilde{\xi}_{\mathcal{T}{2}} (\text{fSim}(\theta,\phi)) > \xi (\text{fSim}(\theta,\phi))$. That is, we fail to find the correct value of the stabilizer extent if we restrict to expansions over the set $\Cl{n}^{\T{2}}$. Fig.~\ref{fig: xi fSim gates}(b) provides a clear visualization of the results.

The fSim gates are also invariant under qubit permutations. We work similarly to the case of transposition invariance and consider decompositions over the permutation-invariant Clifford subgroup $\Cl{2}^{\Sn{2}}$ with associated stabilizer extent $\tilde{\xi}_{\Sn{2}}(\text{fSim}(\theta,\phi))$.
In this case, the results are even worse than the ones obtained with the transposition-invariant Clifford gates; in particular, while the optimization problem could always be solved, the agreement between $\xi$ and $\tilde{\xi}_{\Sn{2}}$ is restricted to a smaller parameter region and the ratio $\tilde{\xi}_{\Sn{2}}/\xi$ achieves much larger values (see Fig.~\ref{fig: xi fSim gates}(c)).

\section{Applications}\label{sec: Practical demonstrations}

\subsection{Details of the practical implementation}\label{subsec: Implementation details}

As previously explained, this work aims to improve the solution of a particular SOCP. To this end, and before the optimization itself, we have to build the set of elements to be considered in the main constraint given to the optimizer. To generate the group of Clifford unitaries, we used the algorithm provided by Koenig and Smolin~\cite{koenig2014efficiently}. On the other hand, to generate the diagonal and real-diagonal subgroups, we constructed our own algorithm~\cite{GitRepo}. Both of these algorithms exploit a particular emergent structure from the isomorphism between $\Cl{n}/\mathcal{P}_n$ and $\Sp (2n)$, where $\Sp(2n)$ refers to the symplectic group of $2n\times2n$ matrices (see, for example, Ref.~\cite{rengaswamy2018synthesis}).

In our code, Lemma~\ref{lemma: weak reduction} is leveraged to reduce the size of the search space of the optimization problem for diagonal and real-diagonal unitaries that are also permutation-invariant. This is done by directly filtering the elements of the starting set of diagonal and real-diagonal Cliffords. Finally, to solve the optimization problem itself, we employed the third-party software Gurobi~\cite{gurobi} on a laptop with an Intel(R) Core(TM) i5-155U processor and 32GB of RAM.

We verified our code by computing the stabilizer extent of $T^{\otimes k}$ for $k\in\{1,\cdots,6\}$, $CS$, $CS^{\otimes 2}$, $CCZ$, and $CCZ^{\otimes 2}$ and comparing the results against the known theoretical values by Bravyi \emph{et al.}~\cite{Bravyi2019}. All of the instances yielded the correct outcomes. In the following sections, we present novel results for other important unitaries.

Our code is openly available on GitHub~\cite{GitRepo}.

\subsection{\texorpdfstring{Characterization of multicontrolled phase gates}{Characterization of multicontrolled phase gates}}\label{subsection: multicontrolled-Z-rotations}

We start by calculating the stabilizer extent of the family of multicontrolled phase gates on $n$ qubits, denoted $C^{n-1} P (\theta)$ and defined as
\begin{equation}\label{eq: def multicontrolled Z}
    C^{n-1} P (\theta) = \mathrm{diag}\left( 1,1,\dots,1,e^{i\theta} \right) \in \mathbb{C}^{2^n}\,.
\end{equation}

\begin{table*}[t]
\subfloat[$C^{n-1}Z$ gates.]{%
\begin{tabular}{c c c}
        \hline\hline
        $n$ & Gate & $\xi$ \\
        \hline
        1 & $Z$ & 1\\
        2 & $CZ$ & 1\\
        3 & $CCZ$ & 1.77778\\
        4 & $C^3Z$ & 2.25000\\
        5 & $C^4Z$ & 2.50694\\
        6 & $C^5Z$ & 2.64063\\
        7 & $C^6Z$ & 2.70877\\
         \hline\hline
\end{tabular}
}\hspace{1cm}
\subfloat[$C^{n-1}S$ gates.]{%
\begin{tabular}{c c c}
        \hline\hline
        $n$ & Gate & $\xi$ \\
        \hline
        1 & $S$ & 1\\
        2 & $CS$ & 1.6\\
        3 & $C^2S$ & 2.05000\\
        4 & $C^3S$ & 2.31250\\
        5 & $C^4S$ & 2.45313\\
        6 & $C^5S$ & 2.52578\\
        7 & $C^6S$ & ---\\
         \hline\hline
\end{tabular}
}\hspace{1cm}
\subfloat[Gate with maximal stabilizer extent, for each $n$.]{%
\begin{tabular}{c c c c}
        \hline\hline
        $n$ & Gate & $\theta_{\text{max}}$ & $\xi$ \\
        \hline
        1 & $T$    & $\pi/4$ & $1.17157 $\\
        2 & $CS$   & $\pi/2$ & 1.6 \\
        3 & $C^2P(\theta_{\text{max}})$ & $0.6476\pi$ & $2.13263$ \\
        4 & $C^3P(\theta_{\text{max}})$ & $0.7048\pi$& $2.50000$ \\
        5 & $C^4P(\theta_{\text{max}})$ & $0.7288\pi$& $2.70792$ \\
        6 & $C^5P(\theta_{\text{max}})$ & $0.7397\pi$ & $2.81774$ \\
        7 & $C^6P(\theta_{\text{max}})$ & --- & ---\\
         \hline\hline
\end{tabular}\label{subtab: Maximal extent of multicontrolled phase gates}
}\hspace{1cm}
\caption{\textbf{Stabilizer extent, $\xi$, for relevant gates.} We present the value of the stabilizer extent for the (a) $C^{n-1}Z$ gates, (b) $C^{n-1}S$ gates, and (c) gates with maximal stabilizer extent for each $n$. The latter results were obtained by interpolation of the values used to construct the curves in Fig.~\ref{fig: CRZ(th) with synthesis}.}\label{tab: Summary of extent results CZ, CS, maximal}
\end{table*}

The results obtained for the stabilizer extent of this family of gates are presented in Fig.~\ref{fig: CRZ(th) with synthesis}. We see that, as the number of (control) qubits increases, the distance between the curves starts to diminish. Moreover, the position of the absolute maxima progressively moves towards $\theta = \pi$ and the local minimum present at that position seems to become less pronounced. We conjecture that, in the limit $n \rightarrow\infty$, the local minimum disappears and the curve has a single maximum at $\theta_{\text{max}} = \pi$.

\begin{figure}[t]
    \centering
    \includegraphics[width=0.99\linewidth]{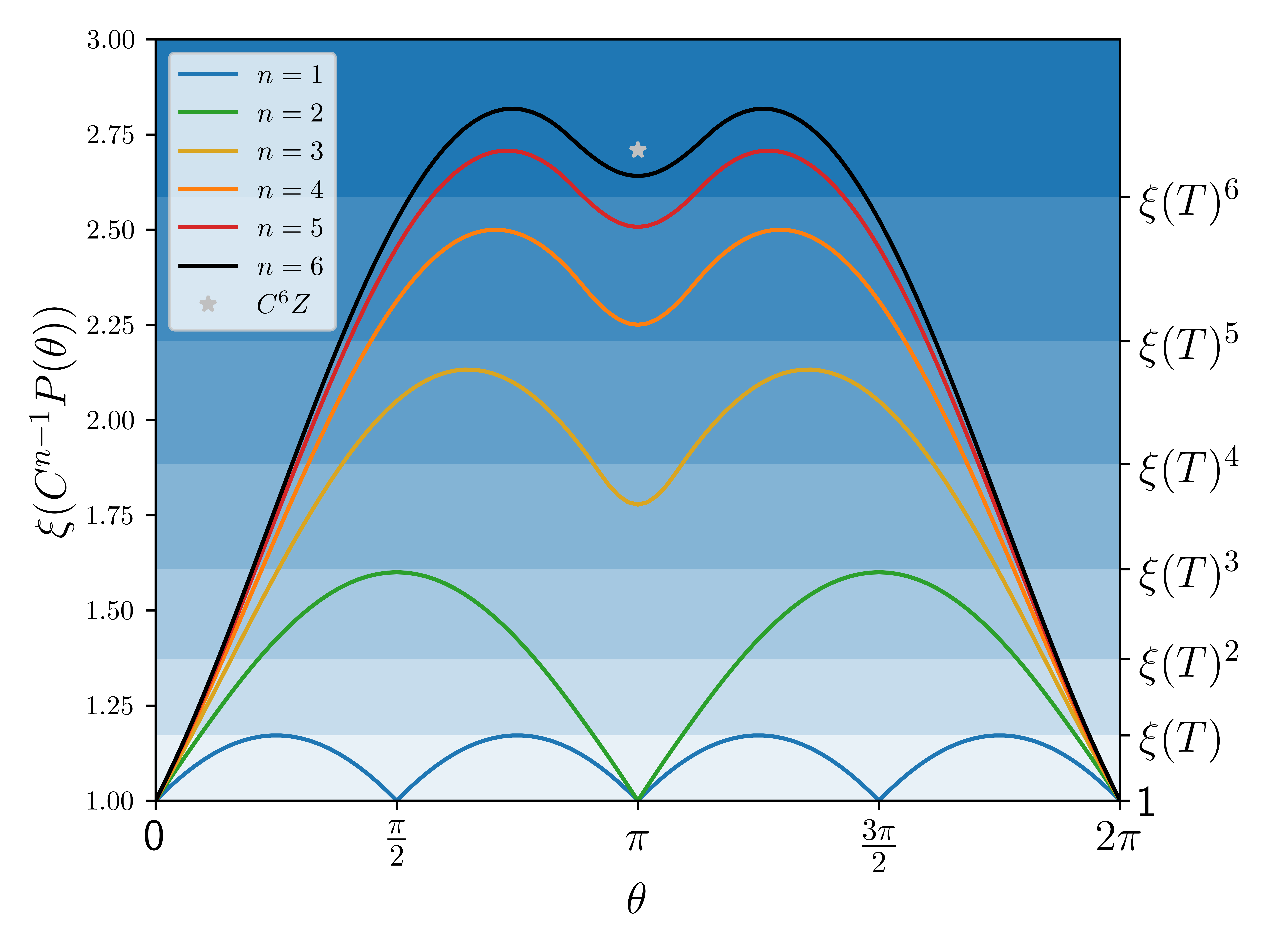}
    \caption{\textbf{Stabilizer extent of the multicontrolled-$P(\theta)$ gates as a function of $\theta$.} The figure shows the stabilizer extent of the $C^{n-1}P (\theta)$ gates as defined in Eq.~\eqref{eq: def multicontrolled Z} for $n\in \{1,\dots,6\}$. The background of the plot is divided into regions colored in different shades of blue, highlighting values of the stabilizer extent between $\xi (T^{s-1})$ and $\xi (T^s)$ for $s\in \{1, 7\}$ and allow us to immediately identify the minimum number of $T$ gates needed to synthesize the different gates, as per Proposition~\ref{prop: gate synthesis}.}
    \label{fig: CRZ(th) with synthesis}
\end{figure}

In Table~\ref{tab: Summary of extent results CZ, CS, maximal}, we present the explicit values of the stabilizer extent for the $C^{n-1}Z$ and $C^{n-1}S$ gates. We note that, up to $n=4$, $\xi (C^{n-1}Z) < \xi (C^{n-1}S)$, but for $n\geq 5$ this relation is reversed. This is consistent with the previous conjecture that in the limit $n \rightarrow\infty$ the maximum occurs at $\theta_{\text{max}} = \pi$. Additionally, the table also presents the maximum extent value associated with the gate family $C^{n-1} P (\theta)$ for each value of $n$ and the corresponding rotation angle $\theta_{\max}$.

We use this set of gates to make interesting studies concerning the properties of submultiplicativity of the stabilizer extent under tensor product and composition. We go further by analyzing the role of entanglement in the decompositions of this important family of unitaries and conclude with comments on the cost of synthesizing some of these gates.

\subsubsection{(Sub)Multiplicativity of the stabilizer extent under tensor product and composition}

\begin{figure}[b]
    \centering
    \includegraphics[width=1\linewidth]{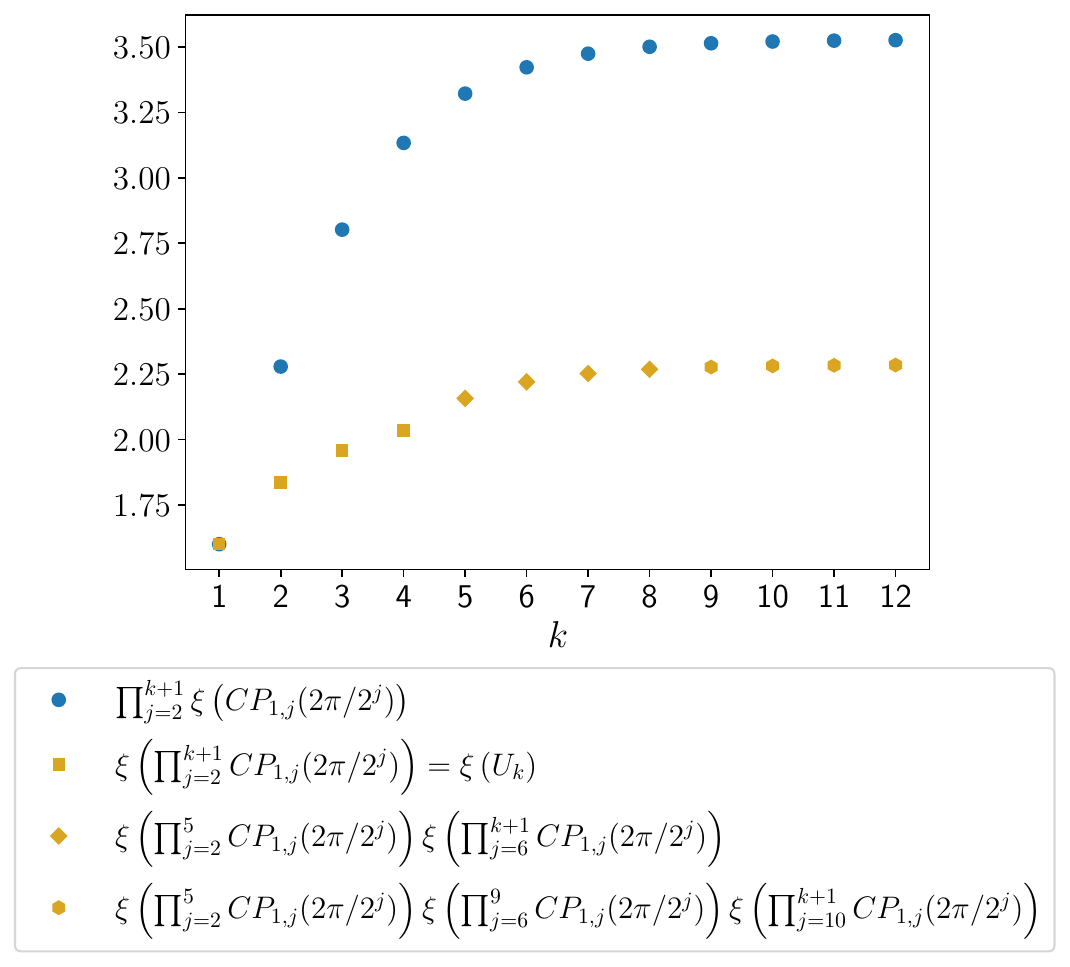}
    \caption{\textbf{Submultiplicativity of the controlled-phase gates in the quantum Fourier transform subroutine.} The blue circles depict the values obtained by multiplying the stabilizer extent of the individual $\{CP_{1,j}(2\pi / 2^j)\}_{j=2}^{k+1}$ gates. Each circle corresponds to the value obtained for the largest unitary block, $U_k$, in a QFT subroutine with $(k+1)$ qubits. In yellow, we present the results obtained when splitting each $U_k$ into unitary blocks of the maximum size. Up to $k=4$, the results correspond to the true value of the stabilizer extent. For $5\leq k \leq 8$, each $U_k$ is split into two blocks whose stabilizer extent is multiplied. Finally, for $k\geq 9$, each $U_k$ is split into three blocks.}
    \label{fig: QFT submultiplicativity}
\end{figure}

We numerically tested the (sub)multiplicativity of the stabilizer extent under tensor products and multiplicativity holds for all the tested instances. Namely, up to five tensor products of the same single-qubit $P(\theta)$, as well as tensor products of $CP (\theta)$ with other $CP (\theta)$ or the $C^2P (\theta)$ gate, and also of $P(\theta)$ and $C^3P(\theta)$, for $\theta\in[0,2\pi]$; that is to say, we tested the multiplicativity of the curves appearing in Fig.~\ref{fig: CRZ(th) with synthesis}. More generally, we also considered tensor products with elements having different angles randomly selected within that interval, including $CP(\theta_1)$ with $CP(\theta_2)$, $CP(\theta_1)$ with $C^2P(\theta_2)$, $P(\theta_1)$ with $C^3P(\theta_2)$, and $P(\theta_1)$ with multiple other phase gates. As we found no counter-example, the question of whether multiplicativity of the stabilizer extent holds for 1-, 2-, and 3-qubit unitaries (as it does for quantum states) remains open.

Now, a particular subset of the family of multicontrolled phase gates, the $CP(2\pi / 2^j)$ gates, is used in the textbook formulation of the quantum Fourier transform (QFT) subroutine. We focus on these gates to analyze the submultiplicativity of the stabilizer extent under composition. In particular, we consider unitary blocks of the form:
\begin{equation*}
    U_k = \prod_{j=2}^{k+1} CP_{1,j} (2\pi / 2^j) \,,
\end{equation*}
where the subscript indicates that the two-qubit gate is controlled in the first qubit and targets the $j$th qubit. Hence, each $U_k$ is a $(k+1)$-qubit unitary with $k$ controlled-phase gates. 

The results are depicted in Fig.~\ref{fig: QFT submultiplicativity}. We consider different values of $k$ ranging from 2 to 12. The blue circles represent the values obtained by computing the stabilizer extent of each two-qubit gate in $U_k$ and multiplying the results. Contrastingly, the yellow shapes (squares, diamonds, and octagons) correspond to the values obtained by splitting each $U_k$ into unitary blocks of 5 qubits (the largest blocks whose stabilizer extent can be computed using our code), calculating the stabilizer extent for each block, and multiplying the results. The figure perfectly captures the strict submultiplicativity of the stabilizer extent.

These results will be useful in Sec.~\ref{subsec: QFT}, when we analyze the remarkable improvements that our results provide to the simulation of QFT circuits using the naive sum-over-Cliffords approach.

\subsubsection{The role of entanglement}

We now study how entanglement impacts the value of the (squared) $\ell_1$-norm of stabilizer decompositions of $C^{n-1}Z$ gates. Because these gates are real and diagonal, Theorem~\ref{theorem: Main result -- diagonals and reals} guarantees that an optimal decomposition can be found in terms of real-diagonal Clifford unitaries. Every $C\in \RD{n}$ can be written as
\begin{equation*}
    C = \prod_{j=1}^n Z_j^{a_j} \prod_{k>j}^n CZ^{b_{jk}}_{jk}\,,
\end{equation*}
where $a_j \in \{0,1\}$ determines whether a $Z$ gate is acting on the $j$th qubit and $b_{jk} \in \{0,1\}$ establishes whether a $CZ$ gate is being applied between qubits $j$ and $k$. It is therefore easy to see that the number of $CZ$ gates in a given $n$-qubit real-diagonal Clifford unitary is bounded between 0 and $n(n-1)/2\,.$

To gain a deeper understanding of the entangling structures underpinning the optimal decompositions of $C^{n-1} Z$ gates, we break down the real-diagonal Clifford unitaries into sets defined by the number of $CZ$ gates. We denote these sets $\mathfrak{C}_i$, where the subscript $i$ denotes the number of $CZ$ gates of the real-diagonal Clifford unitaries in the set. For example, $\prod_{j=1}^n Z_j^{a_j} \in \frak{C}_0$, every $\prod_{j=1}^n Z_j^{a_j} CZ_{kl}$ with $l>k$ is in $\frak{C}_1$, every $\prod_{j=1}^n Z_j^{a_j} CZ_{kl} CZ_{mp}$ with $l>k$, $p>m$, and $(k,l)\neq (m,p)$ is in $\frak{C}_2$, and so on. After doing this separation of $\RD{n}$ into subsets $\{\frak{C}_i\}_i$, we can investigate how the square of the $\ell_1$-norm of the best stabilizer decomposition of $C^{n-1}Z$ gates changes when we narrow the search space for the linear program from the whole $\RD{n}$ group to specific collections of $\frak{C}_i$'s. In doing this additional (heuristic) reduction of the search space, we are not guaranteed to find the absolute optimal decomposition and are likely to uncover only a suboptimal one. Remarkably, our results show that, for $C^{n-1} Z$ gates with $n>3$, the optimal decomposition can be found within a restricted collection of $\frak{C}_i$ sets, highlighting that while some entangling structures are strictly necessary, others are not.

In Fig.~\ref{fig: entanglement role}, we show how our figure of merit changes as we increase the number of sets $\mathfrak{C}_i$ that are being used during the optimization. As we can see, for $ 3 < n \leq 7$, we do not need to consider all $\mathfrak{C}_i$ in order to achieve the optimal solution: $\xi(C^{n-1}Z)$. We conjecture that the same will hold also for $n>7$, suggesting that further heuristic reductions to the search space might be possible, potentially enabling the computation of the stabilizer extent for even larger unitaries. A finer look into Fig.~\ref{fig: entanglement role} reveals that, for $n\leq 6$, four different $\mathfrak{C}_i$ sets suffice to achieve the optimal decomposition of $C^{n-1}Z$ gates. On the other hand, for $n=7$, five different $\mathfrak{C}_i$ sets are needed. Nevertheless, this is still a large reduction when we consider the total of 22 different sets $\{\mathfrak{C}_i\}_{i=0}^{21}$ that exist for that number of qubits.

We complement the analysis above with the information in Table~\ref{tab: sext for CCZ using different categories of CZ}, which explicitly presents the smallest $\ell_1$-norm squared obtained when different $\mathfrak{C}_i$ sets are used to decompose the $CCZ$ gate.

\begin{figure}[t]
    \centering
    \includegraphics[width=0.75\linewidth]{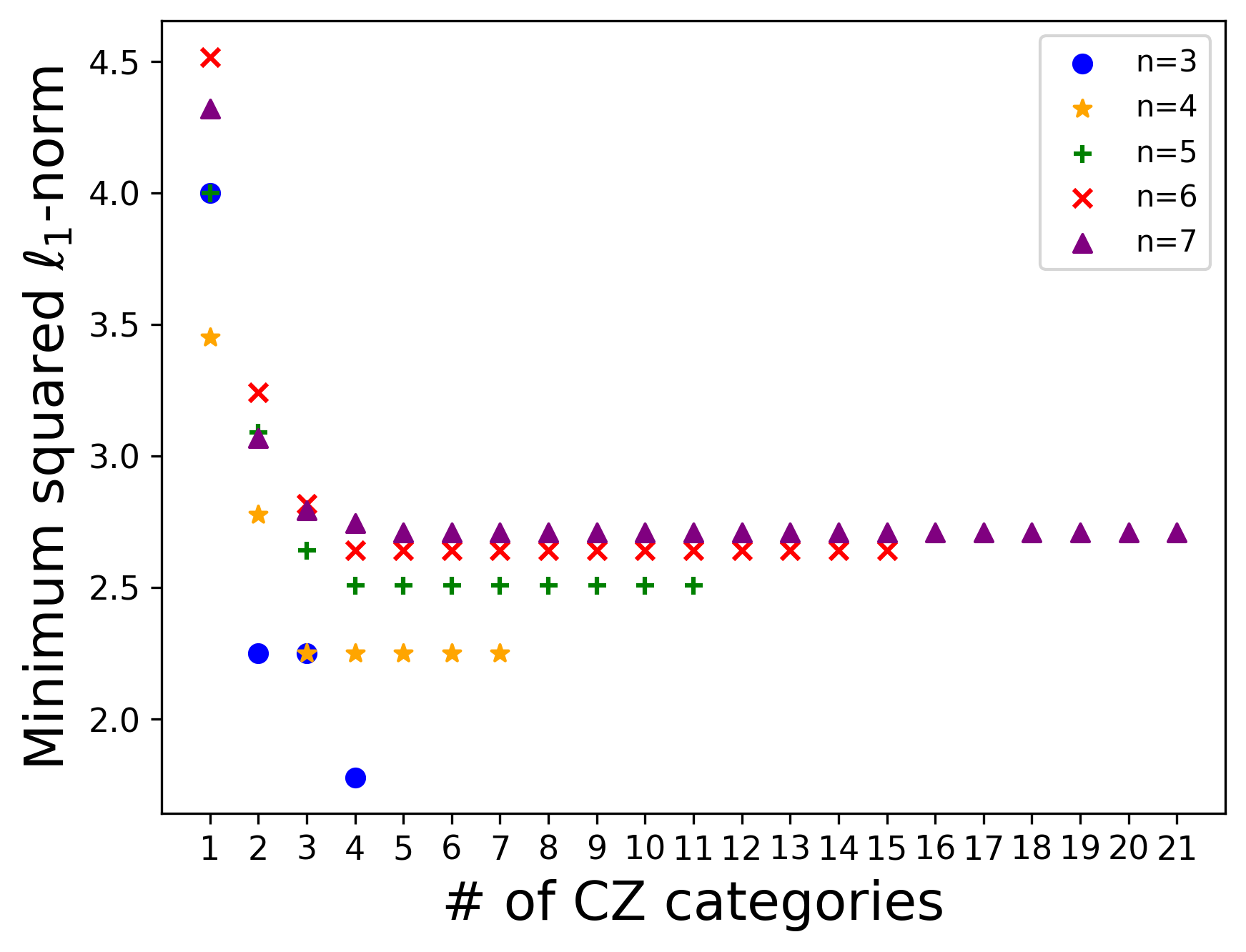}
    \caption{\textbf{The role of entanglement in the computation of $\xi(C^{n-1}Z)$.} We show the minimum value found for the value of the $\ell_1$-norm squared of the decomposition of $C^{n-1} Z$ gates in terms of real-diagonal Cliffords. This value is plotted against the number of $CZ$ sets, $\mathfrak{C}_i$, used during the optimization. Note that given a fixed number of sets, several different combinations are possible; we tested all of these combinations and present only the minimum value in each case. Note that, for $n\leq 6$, the value of the stabilizer extent (i.e., the absolute minimum $\ell_1$-norm squared) is achieved using only four $\mathfrak{C}_i$ sets. On the other hand, for $n=7$, five such sets are needed.}
    \label{fig: entanglement role}
\end{figure}

\begin{table}
    \centering
    \begin{tabular}{c c}
        \hline\hline
        $\mathfrak{C}_i$ sets used during optimization & Min. $\ell_1$-norm squared \\
         \hline
        $\mathfrak{C}_0$ or $\mathfrak{C}_3$ & 6.25000\\
        $\mathfrak{C}_1$ or $\mathfrak{C}_2$ & 4.00000\\
        $\mathfrak{C}_0 \cup \mathfrak{C}_2$ or $\mathfrak{C}_1 \cup \mathfrak{C}_3$ & 4.00000\\
        $\mathfrak{C}_0 \cup \mathfrak{C}_3$ & 4.00000\\
        $\mathfrak{C}_0 \cup \mathfrak{C}_1$ or $\mathfrak{C}_2 \cup \mathfrak{C}_3$ & 2.56000\\
        $\mathfrak{C}_1 \cup \mathfrak{C}_2$ & 2.25000\\
        any $\mathfrak{C}_i \cup \mathfrak{C}_j \cup\mathfrak{C}_k$ with $i\neq j\neq k \neq i$ & 2.25000\\
        $\cup_{i=0}^3 \mathfrak{C}_i$ & 1.77777 (optimal)\\
         \hline\hline
    \end{tabular}
    \caption{\textbf{Minimum $\ell_1$-norm squared of the decomposition of the $CCZ$ gate, using different combinations of $CZ$ sets.} For $n=3$, diagonal-real Clifford operators can be separated into four $CZ$ sets, $\{\mathfrak{C}_i\}_{i=0}^3$, where $i$ fixes the number of $CZ$ gates of the unitaries in each set. We present the results from the largest to the smallest best possible value of the square of the $\ell_1$-norm. In this case, the absolute optimal value is reached only when all three sets $\{\frak{C}_i\}_i$ are used.}
    \label{tab: sext for CCZ using different categories of CZ}
\end{table}

We conclude this discussion by pointing out that, since the $C^{n-1}Z$ gates are not only real-diagonal but also permutation-invariant, we also leveraged Lemma~\ref{lemma: weak reduction} to further reduce the memory requirements. For details, the interested reader is referred to Appendices~\ref{app: cardinality of Dn/Sn} and \ref{app: Permutation invariance}.

\subsubsection{Statements concerning gate synthesis}

Combining Proposition~\ref{prop: gate synthesis} with the enhanced ability to calculate the stabilizer extent of unitaries, we improve our capacity to lower bound the number of $T$ gates needed to synthesize a given unitary. Fig.~\ref{fig: CRZ(th) with synthesis} allows us to immediately identify what is the minimum number of $T$ gates that are needed for synthesizing each gate in the family of multicontrolled phase gates. For instance, we see that $CS$ requires at least 3 $T$ gates, the $CCZ$ at least 4, and the $CCS$ requires at least 5. These results are consistent with the ones obtained in Ref.~\cite{Howard2017} using the robustness of magic. However, in that work, the authors limit their analysis to 3-qubit gates of the third level of the Clifford hierarchy. Here, we easily extend this analysis to a larger set of gates. For instance, we note that the $C^3Z$, $C^3S$, $C^4Z$, and $C^4S$ gates all require at least 6 $T$ magic gates to be synthesized.

Usually, gates with smaller rotation angles require more $T$ gates to be synthesized. This is not reflected by our results, indicating that our lower bound is not tight and that the stabilizer extent might not be the best monotone for drawing conclusions concerning the cost of circuit synthesis. Even so, it provides a good estimate for many gates, as seen by the aforementioned examples of the $CS$ and $CCZ$ gates, which can be exactly synthesized with 3 and 4 $T$ gates, respectively.

\subsection{Speeding up the classical simulation of quantum Fourier transform circuits}\label{subsec: QFT}

\begin{figure*}[t]
    \centering
    \begin{tikzpicture} \node[scale=0.65] { \begin{quantikz}[thin lines] \lstick[wires = 5]{$\ket{0}^{\otimes 5}$} &\gate{H} &\ctrl{1}\gategroup[wires=5,steps=4,style={dashed,rounded corners,fill=blue!20, inner xsep=2pt, inner ysep=2pt},background]{{$U_4$}} &\ctrl{2} &\ctrl{3} &\ctrl{4} &\qw &\qw &\qw &\qw &\qw &\qw &\qw &\qw &\qw &\qw &\qw\\ &\qw &\gate{P(\frac{\pi}{2})} &\qw &\qw &\qw &\gate{H} &\ctrl{1}\gategroup[wires=4,steps=3,style={dashed,rounded corners,fill=red!20, inner xsep=2pt, inner ysep=2pt},background]{{$U_3$}} &\ctrl{2} &\ctrl{3} &\qw &\qw &\qw &\qw &\qw &\qw &\qw\\ &\qw &\qw &\gate{P(\frac{\pi}{4})} &\qw &\qw &\qw &\gate{P(\frac{\pi}{2})} &\qw &\qw &\gate{H} &\ctrl{1}\gategroup[wires=3,steps=2,style={dashed,rounded corners,fill=green!20, inner xsep=2pt, inner ysep=2pt},background]{{$U_2$}} &\ctrl{2} &\qw &\qw &\qw &\qw\\ &\qw &\qw &\qw &\gate{P(\frac{\pi}{8})} &\qw &\qw &\qw &\gate{P(\frac{\pi}{4})} &\qw &\qw &\gate{P(\frac{\pi}{2})} &\qw &\gate{H} &\ctrl{1}\gategroup[wires=2,steps=1,style={dashed,rounded corners,fill=purple!20, inner xsep=2pt, inner ysep=2pt},background]{{$U_1$}} &\qw &\qw \\ &\qw &\qw &\qw &\qw &\gate{P(\frac{\pi}{16})}&\qw &\qw &\qw &\gate{P(\frac{\pi}{8})} &\qw &\qw &\gate{P(\frac{\pi}{4})} &\qw &\gate{P(\frac{\pi}{2})} &\gate{H} &\qw \end{quantikz} }; \end{tikzpicture}
    \caption{\textbf{Example of a 5-qubit circuit implementing the quantum Fourier transform.} Each unitary block $U_j$ has $j$ controlled phase gates. The stabilizer extent of various of these blocks can be seen in Fig.~\ref{fig: QFT submultiplicativity}.}
    \label{fig: QFT circuit}
\end{figure*}
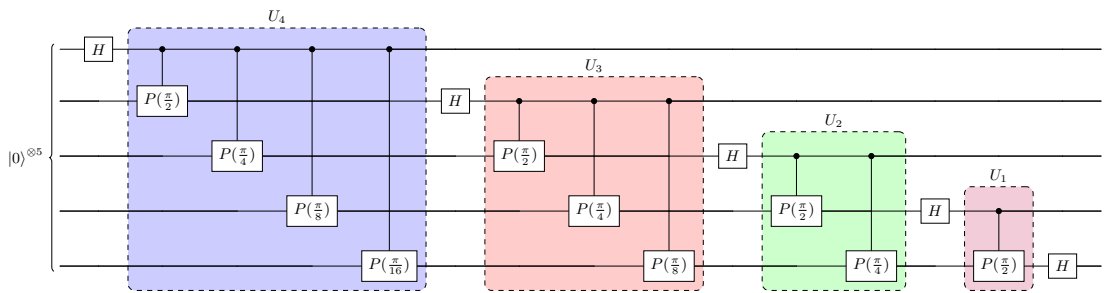

The quantum Fourier transform (QFT) is a fundamental subroutine in many important quantum algorithms, including quantum phase estimation, Shor's factoring algorithm, and quantum support vector machines. While QFT has been shown to be efficiently classically simulable on low-entangled inputs~\cite{Browne2007}, this is no longer the case when embedded as a component within these algorithms. Here, we investigate how the simulation of the QFT block (potentially embedded within an arbitrary quantum algorithm) is significantly enhanced by our results.

For this discussion, we deliberately focus on the runtime savings in simulating the QFT circuit with sum-over-Clifford simulators that stem from our symmetry reduction techniques.
This allows us to illustrate how our methods can be employed in practice to speed up stabilizer-based simulators.
The actual simulation runtime depends on other factors as well, for instance, what input states are used or how the QFT is embedded in a larger circuit. 
This is true for any simulation method (e.g.,~tensor networks) and determines which one should be used in practice.

For simplicity, we start by considering the 5-qubit quantum circuit implementing QFT as seen in Fig.~\ref{fig: QFT circuit}. The unitary blocks $U_1$ to $U_4$ can be decomposed into sums of Cliffords using our Theorem~\ref{theorem: Main result -- diagonals and reals} for diagonal unitaries. The stabilizer extent values obtained for each of these blocks (and previously depicted by the yellow squares in Fig.~\ref{fig: QFT submultiplicativity}) are $\xi (U_1)=1.6$, $\xi (U_2)=1.83566$, $\xi (U_3)=1.95857$, and $\xi (U_4)=2.03480$. This can be used to quantify the cost of simulating such a unitary block: $\prod_{i=1}^4\xi (U_i) \approx
11.705$. Alternatively, up to now, one would be limited to computing the decomposition of each controlled-phase gate individually. In that case, the cost of simulation would be determined by the product of the values of the first four blue points in Fig.~\ref{fig: QFT submultiplicativity}: $\prod_{j=2}^5 \xi (CP(2\pi / 2^j))^{6-j} \approx 32.015$. Hence, owing to the strong submultiplicativity of the stabilizer extent of controlled-phase gates, our results enable speeding up the simulation of this quantum circuit roughly $2.7$ times.

Although this may appear small, the advantage of our method becomes increasingly significant as the number of qubits increases. For instance, for the $8$-qubit QFT circuit, our results enable the use of a (suboptimal) decomposition with a total one-norm squared of $125.25$. Contrastingly, decomposing each non-Clifford controlled-phase gate individually leads to a decomposition with a one-norm squared of $1264.6$, a value that is approximately $10$ times larger. Similarly, for the $16$-qubit QFT circuit, using the commonly used approximation where any controlled-phase with a rotation angle smaller than $2\pi / 2^{13}$ is ignored~\cite{Nam2020AQFT}, our results enable a decomposition with a one-norm squared of $9.2\times 10^{4}$. On the other hand, under the same approximation, the individual decomposition of each controlled-phase gate leads to a one-norm squared of $3.0\times 10^{7}$. Hence, our approach speeds up the simulation of such a circuit roughly 322 times.

Motivated by fault-tolerant computation, quantum circuits are often rewritten into the Clifford$+T$ gate set. We could thus compare the cost of classically simulating the QFT circuits natively using our technique, to that of simulating the equivalent Clifford$+T$ circuits. Based on the results from Refs.~\cite{Nam2020AQFT, Gheorghiu2022}, we expect the $5$-qubit QFT circuit to be approximated by a Clifford+$T$ quantum circuits with 60 to 100 $T$ gates. This corresponds to an associated stabilizer extent between $1.3\times 10^4$ and $7.5\times 10^6$; hence, our approach allows the direct classical simulation of the 5-qubit QFT circuit in Fig.~\ref{fig: QFT circuit} between $1.1\times 10^3$ and $6.4\times 10^5$ times faster than one can simulate the corresponding Clifford+$T$ quantum circuit. Analogously, Nam \emph{et al.}~\cite{Nam2020AQFT} report that their algorithm can synthesize the 8-qubit and 16-qubit QFT circuits, respectively, with 303 $T$ gates and 1162 $T$ gates. This leads to $\xi (T^{\otimes 303}) = 6.9\times 10^{20}$ and $\xi (T^{\otimes 1162}) = 8.1\times 10^{79}$, values that are, respectively, $5.4\times 10^{18}$ and  $8.8\times 10^{74}$ times larger than the values of the one-norm squared of the decompositions obtained leveraging our strong symmetry result.

The task of synthesizing a given unitary into a specific gate set, viz., the Clifford+$T$ gate set, is incredibly hard. Hence, it is reasonable to assume that better algorithms than the one presented by Nam \emph{et al.} could be built. Even still, for the cost of classically simulating these circuits to compete with that of simulating them natively, leveraging our approach, the 5-, 8-, and 16-qubit QFT circuits would have to be synthesized with $T$ counts of 16, 31, and 73, respectively. It seems highly unlikely that such values can be achieved. This highlights an important point, notably that the cost of classically simulating a given quantum computation is sensitive to the set of operations used to express it.

\subsection{Hypergraph quantum states}\label{subsec: Hypergraphs}

We now turn our attention to circuits generating a special type of quantum states known as hypergraph quantum states~\cite{Rossi2013HGs}. Formally, these states are defined as follows.
\begin{definition}[Hypergraph quantum states~\cite{Rossi2013HGs}]\label{def: hypergraph states}
    Let $H = (V,E)$ be an undirected hypergraph where $V$ is the set of vertices and $E$ the set of hyperedges. Then, the corresponding $n$-qubit \emph{hypergraph quantum state} is
    \begin{equation*}
        \ket{\Psi_\mathrm{H}} \coloneqq \prod_{k=1}^n \prod_{e\in E^{(k)}} C^{k-1} Z_e \ket{+}^{\otimes n}\,,
    \end{equation*}
    where $E^{(k)}\subset E$ denotes the set of hyperedges of order $k$, i.e., the set of hyperedges connecting $k$ qubits. For $k=1$, $E^{(1)}=V$.
\end{definition}
 Put into words, given an undirected hypergraph $H$, the corresponding hypergraph state $\ket{\Psi_\mathrm{H}}$ is generated by first associating with each vertex the stabilizer state $\ket{+}$. Then, for every hyperedge $e$, a multicontrolled-$Z$ gate $C^{k-1}Z_e$ is applied to the qubits connected by $e$. We note that graph states are a special type of hypergraph state where all hyperedges are of order $k=2$. 

Hypergraph states are relevant in measurement-based quantum computation (MBQC)~\cite{RaussendorfBriegel2001, Miller2016hierarchy, Takeuchi2019}, quantum error correction~\cite{Lyons2017, ZhuJochym2022topological}, condensed-matter physics~\cite{Miller2016hierarchy, LuGao2019, ZhuJochym2022topological}, quantum metrology~\cite{Gachechiladze2016, Shettell2020}, and quantum foundations~\cite{Gachechiladze2016, Lyons2017, Guhne2014entanglement}. Moreover, these states are intimately connected to important quantum algorithms like the Deutsch-Jozsa and Grover algorithms~\cite{Rossi2013HGs}. Recently, complete categories of four-qubit hypergraph states were experimentally prepared and verified in a reprogrammable silicon-photonic quantum chip~\cite{Huang2024demonstration}.

Additionally, the recent interest in hypergraph quantum states has also unveiled important connections between the resource theory of magic, Boolean function analysis, coding theory, and quantum complexity of phases of matter~\cite{Miller2016hierarchy, LiuWinter2022manybody}. Still, while other non-classical properties of hypergraph states have been thoroughly studied~\cite{Gachechiladze2016, Lyons2017, Guhne2014entanglement, ZhouHamma2022, Sarkar2021squeezingHGS}, the characterization of nonstabilizerness remains fairly unexplored. Ref.~\cite{LiuWinter2022manybody} made a generic study of many-body quantum magic that included interesting results on the nonstabilizerness of the Union Jack state -- a special hypergraph state known to be universal for MBQC~\cite{Miller2016hierarchy}. More recently, Ref.~\cite{Chen2024magicofquantum} studied the nonstabilizerness of hypergraph states via stabilizer Rényi entropies.

\subsubsection{\texorpdfstring{Characterization of $k$-uniform hypergraph states}{Characterization of k-uniform hypergraph states}}

\begin{figure}
    \centering
    \includegraphics[width=0.99\linewidth]{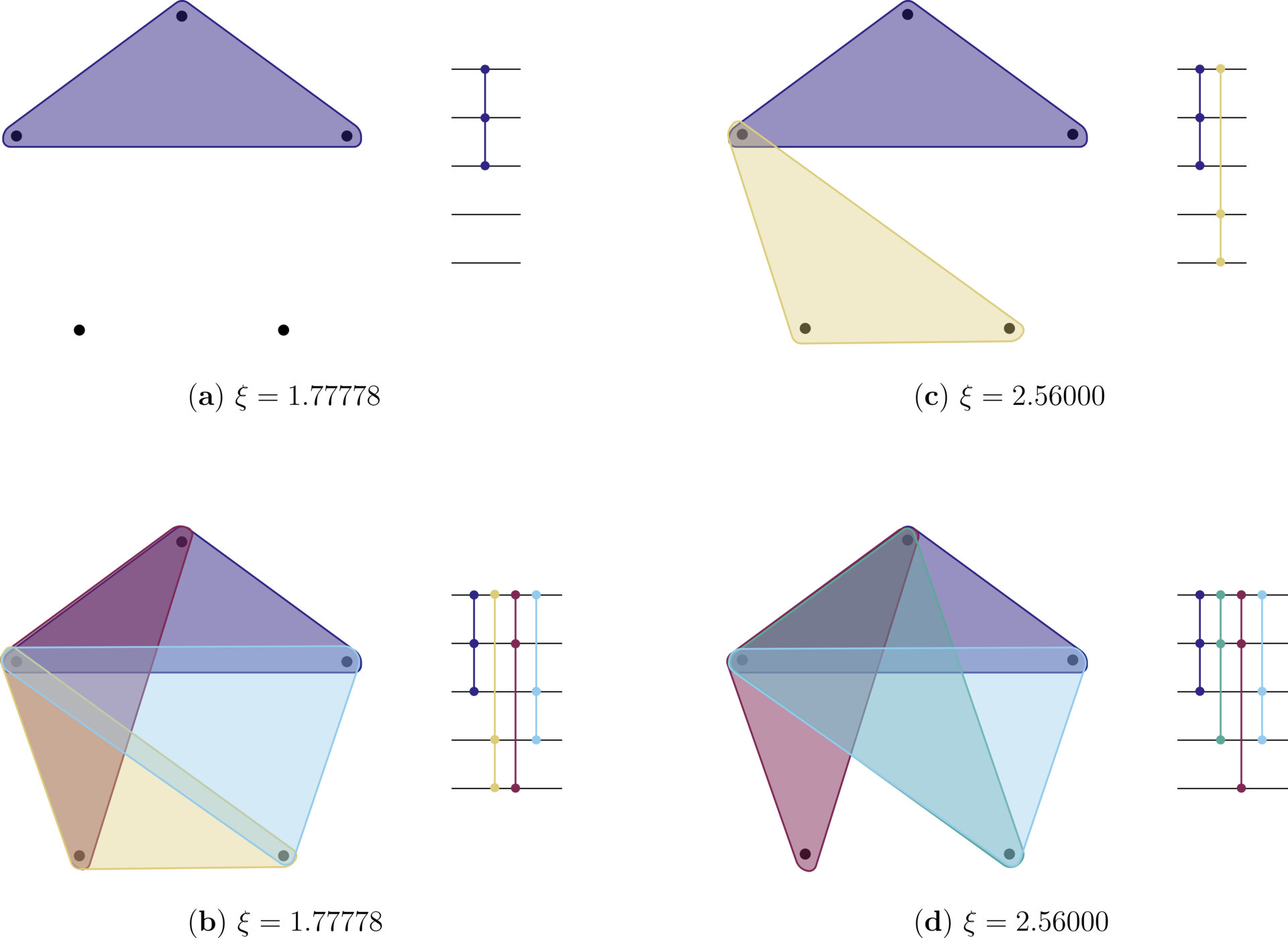}
    \caption{\textbf{Example of four nonequivalent 3-uniform hypergraphs on five vertices.} In this case, there are 33 inequivalent hypergraph states but only 2 distinct values for the stabilizer extent of the unitaries generating them. For each of these two values, we depict only two examples: [(a), (b)] $\xi = 1.77778$ and [(c), (d)] $\xi = 2.56000$. The state depicted in (b) will be important later when we discuss classical simulation of measurement-based quantum computations. Note that (b) and (d) have the same number of hyperedges, but different values of the stabilizer extent.}
    \label{fig: 3-uniform hypergraphs}
\end{figure}

\begin{table}[t]
    \centering
    \begin{tabular}{cccc}
        \hline\hline
        \parbox{16mm}{\# of hyperedges} & \parbox{28mm}{\# of nonequivalent hypergraphs} & \parbox{15mm}{$\xi=$ $1.77778$} & \parbox{15mm}{$\xi=$ $2.56000$}\\
         \hline
        $1$  & $1$ & $1$ & $0$ \\
        $2$  & $2$ & $1$ & $1$ \\
        $3$  & $4$ & $2$ & $2$ \\
        $4$  & $6$ & $2$ & $4$ \\
        $5$  & $6$ & $1$ & $5$ \\
        $6$  & $6$ & $1$ & $5$ \\
        $7$  & $4$ & $1$ & $3$ \\
        $8$  & $2$ & $1$ & $1$ \\
        $9$  & $1$ & $0$ & $1$ \\
        $10$ & $1$ & $0$ & $1$ \\
         \hline\hline
    \end{tabular}
    \caption{\textbf{Distribution of the 33 nonequivalent $3$-uniform hypergraph states of five qubits by number of hyperedges and value of the stabilizer extent of the generating unitary.} The first column fixes the number of hyperedges and the second indicates how many of the nonequivalent hypergraphs have that number of hyperedges. The last two columns indicate how the hypergraphs distribute themselves across the two different values of the stabilizer extent of the generating unitary.}
    \label{tab: extent value 3-uniform 5-qubit hypergraphs}
\end{table}

Here, we leverage our results to compute the stabilizer extent of the unitaries generating five-qubit hypergraph quantum states. Because the number of hypergraphs of five vertices is extremely large, we focus on the study of $k$-uniform hypergraphs, that is, hypergraphs whose hyperedges have fixed order $k$. Since when $k=2$, there is no magic, we start with $k=3.$ The number of nonequivalent $3$-uniform hypergraphs on five vertices is 33. Interestingly, the unitaries that generate the corresponding hypergraph quantum states separate themselves into only two distinct values of the stabilizer extent: $\xi = 1.77778$ and $\xi = 2.56000$. In Fig.~\ref{fig: 3-uniform hypergraphs}, we depict four examples of 3-uniform hypergraphs on five vertices, the circuit generating the corresponding hypergraph state from the state $\ket{+}^{\otimes 5}$, and the associated stabilizer extent. We also highlight how hypergraph states with the same number of hyperedges may exhibit different values of the stabilizer extent of the underlying unitary (cf.~Fig.~\ref{fig: 3-uniform hypergraphs}(b) and~\ref{fig: 3-uniform hypergraphs}(d) and Table~\ref{tab: extent value 3-uniform 5-qubit hypergraphs}).

Fig.~\ref{fig: 4-uniform hypergraphs} depicts all 5 nonequivalent 4-uniform hypergraphs on 5 qubits, as well as the circuit that generates the corresponding hypergraph quantum state from state $\ket{+}^{\otimes 5}$ and the associated value of the stabilizer extent. In this case, there are three distinct values for this quantity: $\xi = 2.25000$, $\xi=2.56000$, and $\xi=2.77778$.

\begin{figure*}
    \centering
    \includegraphics[width=0.8\linewidth]{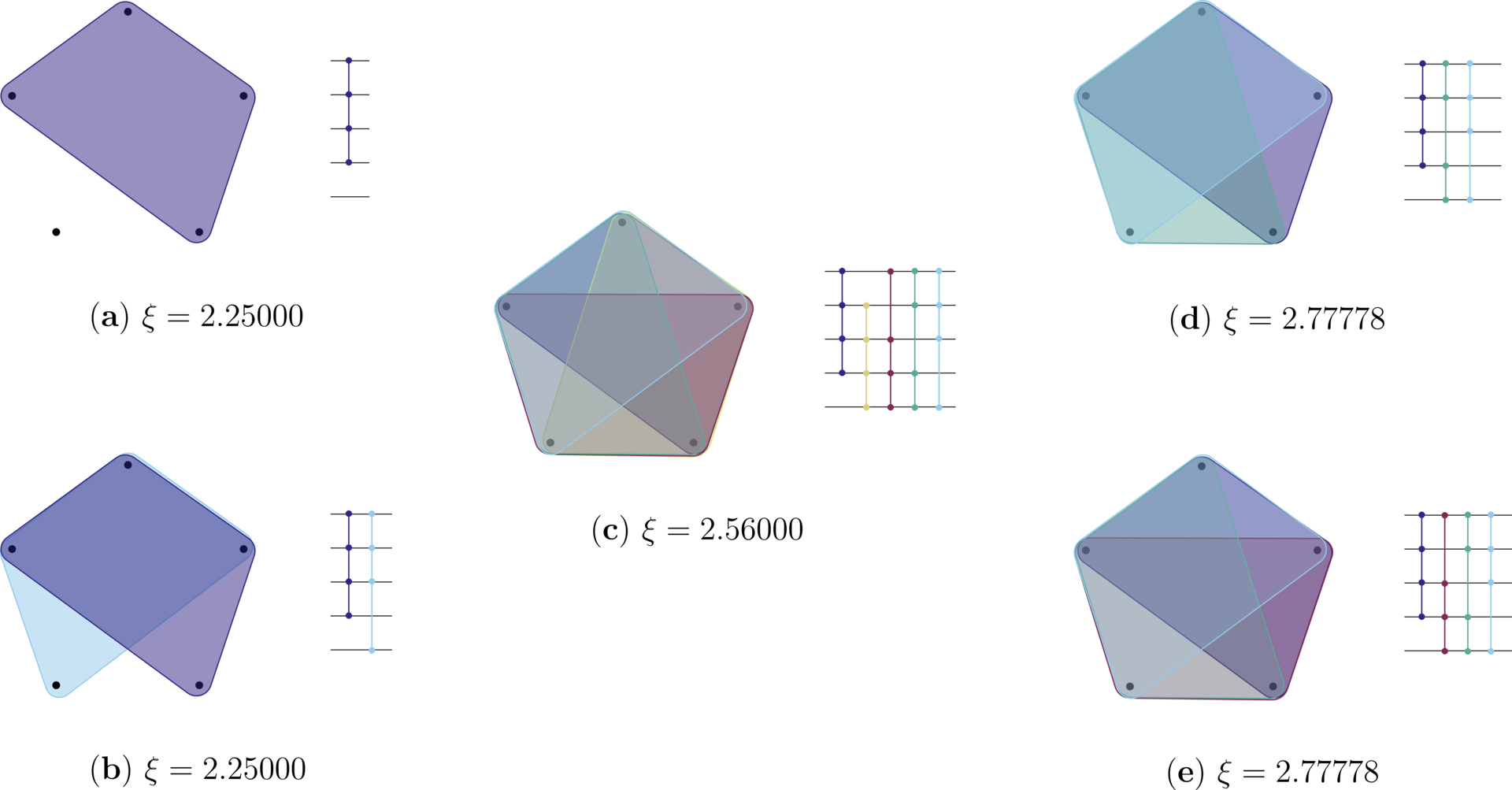}
    \caption{\textbf{All five nonequivalent 4-uniform hypergraphs on five vertices.} We see that there are only three distinct values for the stabilizer extent for the unitaries generating the corresponding hypergraph states: $\xi=2.25000$ when we have one (a) and two (b) $C^{(3)}Z$ gates, $\xi = 2.56000$ when all $C^{(3)}Z$ are present (c), and $\xi = 2.7778$ when there are three (d) and four (e) $C^{(3)}Z$ gates.}
    \label{fig: 4-uniform hypergraphs}
\end{figure*}

\subsubsection{Speeding up the classical simulation of (universal) measurement-based quantum computation}

The Union Jack lattice depicted in Fig.~\ref{fig: Union Jack lattice} is known to be universal for MBQC with Pauli measurements~\cite{Miller2016hierarchy}. Thus, any quantum computation can be written as a circuit on $n$ qubits initialized in the state $\ket{+}^{\otimes n}$, followed by the application of the magic and entangling unitary that generates the Union Jack state, and Pauli measurements at the end. Such a computation can be simulated using the sum-over-Cliffords approach. If we were to use the $CCZ$ decomposition used by Bravyi \emph{et al.}~\cite{Bravyi2019}, the simulation time would grow as $\mathcal{O}(2^{1.6601n})$. Alternatively, our result allows us to decompose each yellow cell in Fig.~\ref{fig: Union Jack lattice} directly, which leads to an associated simulation runtime of $\mathcal{O}(2^{0.4150n})$. This corresponds to an exponential speed up of $\mathcal{O}(2^{1.2451n})$ to the classical simulation of any measurement-based quantum computation written in terms of the Union Jack state with Pauli measurements.

\begin{figure}[b]
    \centering
    \includegraphics[width=0.5\linewidth]{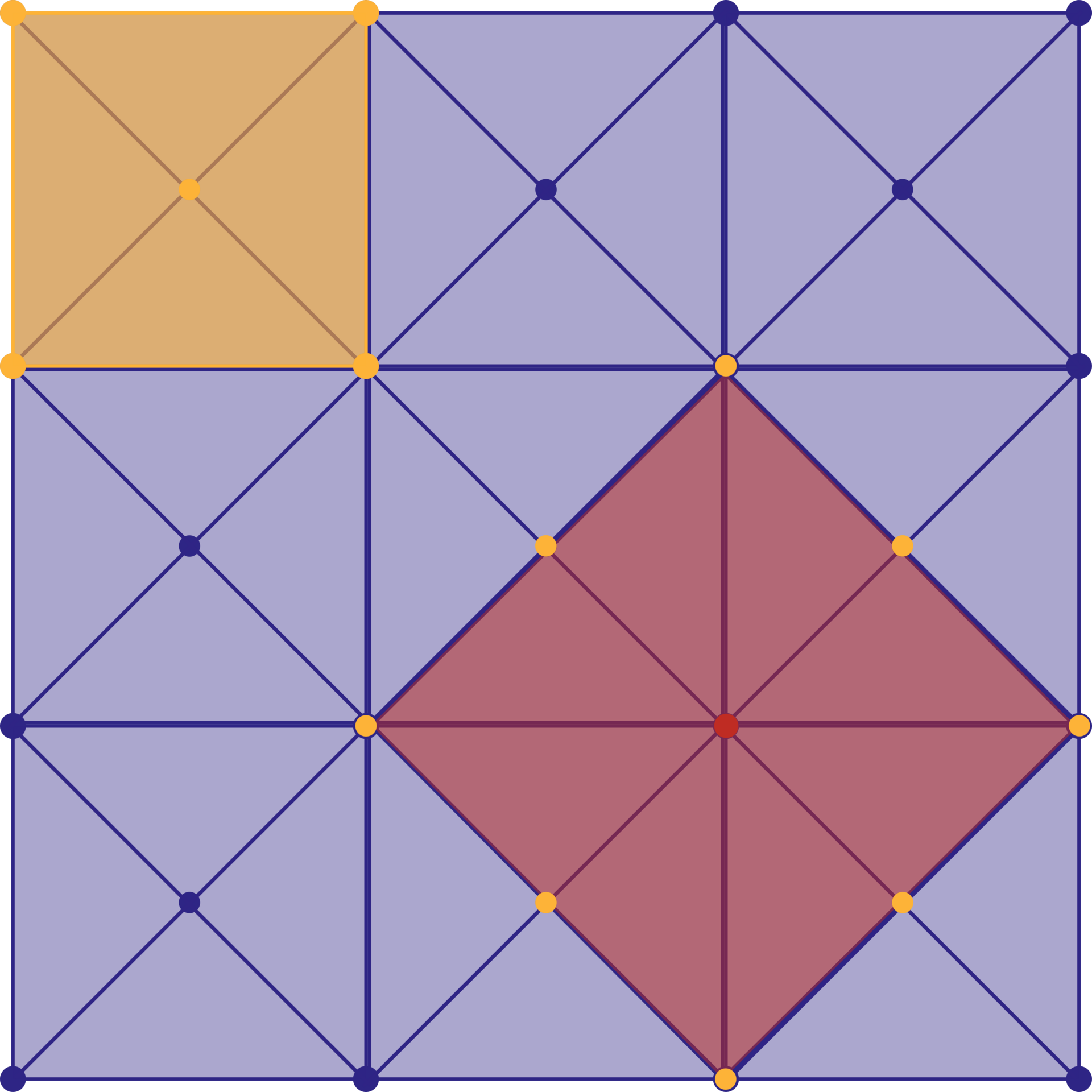}
    \caption{\textbf{Union Jack lattice.} The red shaded area represents the unit cell of the lattice. The yellow shaded square represents the 3-uniform 5-qubit hypergraph state depicted in Fig.~\ref{fig: 3-uniform hypergraphs}(b), which can also be used to generate the entire Union Jack lattice.}
    \label{fig: Union Jack lattice}
\end{figure}

From Ref.~\cite{LiuWinter2022manybody}, it follows immediately that $\xi (\ket{\Psi_{\text{UJ}}}) < 2^{0.46n}$, where $\ket{\Psi_{\text{UJ}}}$ denotes the Union Jack hypergraph quantum state. Hence, our results also provide a better upper bound for the stabilizer extent of the Union Jack state:
\begin{equation*}
    \xi ( \ket{\Psi_{\text{UJ}}} ) < \xi (U_{\text{UJ}}) = \mathcal{O}(2^{0.4150n})\,,
\end{equation*}
with $\ket{\Psi_{\text{UJ}}} = U_{\text{UJ}} \ket{+}^{\otimes n}$.

\subsection{Generalized hypergraph quantum states}\label{subsec: Generalized hypergraphs}

\begin{figure*}[ht]
    \centering
    \begin{tikzpicture} \node[scale=0.65] { \begin{quantikz}[thin lines] \lstick[wires = 5]{$\ket{+}^{\otimes 5}$} &\gate{T}\gategroup[wires=5,steps=1,style={dashed,rounded corners,fill=blue!20, inner xsep=2pt, inner ysep=2pt},background]{{$T^{\otimes 5}$}} &\ctrl{1}\gategroup[wires=5,steps=6,style={dashed,rounded corners,fill=red!20, inner xsep=2pt, inner ysep=2pt},background]{{$U^{E^{(2)}}_{\max}$}} &\ctrl{2} &\ctrl{3} &\ctrl{4} &\qw &\qw &\ctrl{1}\gategroup[wires=5,steps=2,style={dashed,rounded corners,fill=green!20, inner xsep=2pt, inner ysep=2pt},background]{{$U^{E^{(3)}}_{\max}$}} &\ctrl{3} &\ctrl{1}\gategroup[wires=5,steps=3,style={dashed,rounded corners,fill=purple!20, inner xsep=2pt, inner ysep=2pt},background]{{$U^{E^{(4)}}_{\max}$}} &\ctrl{1} &\ctrl{1} &\ctrl{1}\gategroup[wires=5,steps=1,style={dashed,rounded corners,fill=gray!20, inner xsep=2pt, inner ysep=2pt},background]{{$U^{E^{(5)}}_{\max}$}} &\qw\\ &\gate{T} &\gate{S} &\qw &\qw &\qw &\ctrl{1} &\ctrl{2} &\ctrl{1} &\qw &\ctrl{1} &\ctrl{1} &\ctrl{2} &\ctrl{1} &\qw\\ &\gate{T} &\ctrl{2} &\gate{S} &\qw &\qw &\gate{S} &\qw &\gate{P(0.6476\pi)} &\qw &\ctrl{1} &\ctrl{2} &\qw &\ctrl{1} &\qw\\ &\gate{T} &\qw &\qw &\gate{S} &\qw &\ctrl{1} &\gate{S} &\qw &\ctrl{1} &\gate{P(0.7048\pi)} &\qw &\ctrl{1} &\ctrl{1} &\qw\\ &\gate{T} &\gate{S} &\qw &\qw &\gate{S} &\gate{S} &\qw &\qw &\gate{P(0.6476\pi)} &\qw &\gate{P(0.7048\pi)} &\gate{P(0.7048\pi)} &\gate{P(0.7288\pi)} &\qw \end{quantikz} }; \end{tikzpicture}
    \caption{\textbf{A 5-qubit generalized hypergraph circuit.} For each integer $1 \leq k \leq 5$, we set the multicontrolled phase gates with maximal stabilizer extent, as given by Table~\ref{subtab: Maximal extent of multicontrolled phase gates}. We then construct unitary blocks $U_{\max}^{E^{(k)}}$ associated with $k$-uniform hypergraphs and intended to have large (ideally maximal) stabilizer extent values. The associated results are: $\xi(T^{\otimes5})=2.20723$, $\xi(U_{\text{max}}^{E^{(2)}})=3.76471$, $\xi(U_{\text{max}}^{E^{(3)}})=3.34751$, $\xi(U_{\text{max}}^{E^{(4)}})=3.18877$, and $\xi(U_{\text{max}}^{E^{(5)}})=2.70792$. We compose these blocks as depicted in the figure. The stabilizer extent of the whole circuit is $\xi=3.59638$.}
    \label{fig: generalized hypergraph circuit}
\end{figure*}
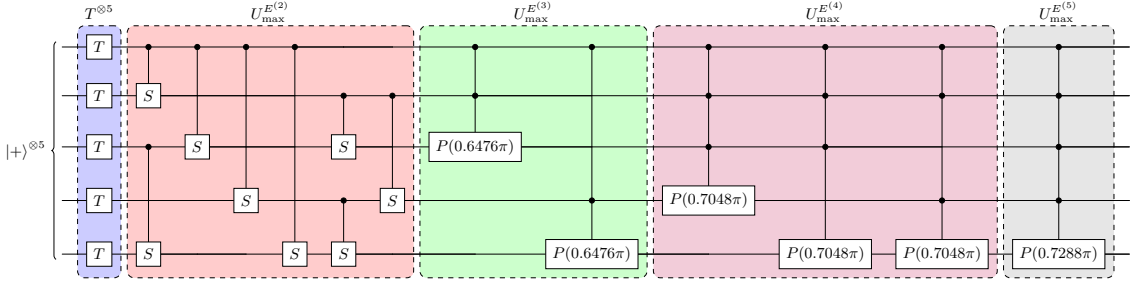

As the name suggests, generalized hypergraph quantum states are a generalization of hypergraph states where the entangling gates are allowed to be multicontrolled phase gates $C^{k-1}P(\theta_k)$ with arbitrary angles, instead of restricting them to be multicontrolled $Z$ gates ($\theta_k=\pi\,\forall k$). Moreover, the angles $\theta_k$ may differ for each order $k$. Formally, we define them as follows.
\begin{definition}[Generalized hypergraph quantum states]\label{def: generalized hypergraph state}
    Let $H=(V,E)$ be an undirected hypergraph, with $V$ its set of vertices and $E$ its set of hyperedges. For each integer $1\leq k \leq \left|V\right|$, set a $(k-1)$-controlled phase gate $C^{k-1}P(\theta_{k})$, with fix angle $\theta_{k}$. Then, the corresponding $n$-qubit \emph{generalized hypergraph quantum state} is
    \begin{equation*}
        \ket{\Psi_{\mathrm{GH}}} \coloneqq \prod_{k=1}^n \prod_{e\in E^{(k)}} C^{k-1} P_{e}(\theta_{k}) \ket{+}^{\otimes n}\,,
    \end{equation*}
    where $E^{(k)}\subset E$ denotes the set of hyperedges of order $k$. For $k=1$, $E^{(1)}=V$.
\end{definition}

Here, motivated by our previous results, we compute the stabilizer extent of the unitaries generating certain generalized hypergraph states. Specifically, we consider $5$-qubit states and, for each value of $1\leq k\leq 5$, we fix the multicontrolled phase gate to be the one with maximum value of the stabilizer extent (as seen in Table~\ref{subtab: Maximal extent of multicontrolled phase gates}). Then, for $k=1$, the unitary with maximum stabilizer extent is $T^{\otimes 5}$, with a value of $\xi(T^{\otimes5})=2.20723$. For $k=2$, we took all $2$-uniform nonequivalent graphs and computed the extent of the associated unitary obtained by associating with each edge the $CS$ gate (see Appendix~\ref{app: graph states 5 vertices} for the complete set of results). We then selected the unitary with the maximum stabilizer extent, which is denoted $U_{\max}^{E^{(2)}}$ and has a stabilizer extent of $\xi(U_{\text{max}}^{E^{(2)}})=3.76471$. For $k=3$ and $k=4$, we took the $3$- and $4$-uniform hypergraphs in Figs.~\ref{fig: 3-uniform hypergraphs} and \ref{fig: 4-uniform hypergraphs} that were associated with the largest value of the stabilizer extent when using multicontrolled Z gates, and replaced those gates by the $C^2P(0.6476\pi)$ and $C^3P(0.7048\pi)$ gates, respectively. We denote the resulting unitaries, respectively, $U^{E^{(3)}}_{\max}$ and $U^{E^{(4)}}_{\max}$ and their stabilizer extents are $\xi(U_{\text{max}}^{E^{(3)}})=3.34751$ and $\xi(U_{\text{max}}^{E^{(4)}})=3.18877$. As expected, these results are larger than the ones obtained with $C^2 Z$ and $C^3 Z$. Finally, for $k=5$, we straightforwardly chose $U_{\max}^{E^{(5)}} = C^4P(0.7288\pi)$ for which $\xi(U_{\text{max}}^{E^{(5)}})=2.70792$ (cf.~Table~\ref{subtab: Maximal extent of multicontrolled phase gates}).

The unitaries $U_{\max}^{E^{(k)}}$ can be used to construct a quantum circuit as depicted in Fig.~\ref{fig: generalized hypergraph circuit}. Given how it was constructed, one might be tempted to conjecture that such a circuit would have a very large stabilizer extent or even the largest possible stabilizer extent. However, that is emphatically not the case, as the circuit depicted has $\xi=3.59638 < \xi(U_{\text{max}}^{E^{(2)}})$. Going even further, we study how the stabilizer extent changes under different compositions of the unitary blocks $U_{\max}^{E^{(k)}}$. The results can be seen in Table~\ref{tab: composition_chunks_of_k_uniform_candidates}. We note that some of the combinations have a stabilizer extent lower than that of the whole circuit, while others have a larger value. Remarkably, none of the tested compositions achieves a stabilizer extent larger than that of $U_{\text{max}}^{E^{(2)}}$. As another example, we note that $T^{\otimes 5}$ and $U_{\text{max}}^{E^{(5)}}$ are two blocks which are permutation invariant. However, while removing the former causes a meaningful reduction in the value of the stabilizer extent (from 3.59638 to 3.08596), removing the latter makes the value of the stabilizer increase (to 3.60744).
\begin{table}[ht]
    \centering   
\begin{tabular}{c c}
        \hline\hline
        Composition & $\xi$ \\
         \hline
        $T^{\otimes5}\cdot U_{\text{max}}^{E^{(2)}}$ & $3.08388$\\
        $T^{\otimes5}\cdot U_{\text{max}}^{E^{(2)}}\cdot U_{\text{max}}^{E^{(3)}}$ & $3.37165$\\
        $T^{\otimes5}\cdot U_{\text{max}}^{E^{(2)}}\cdot U_{\text{max}}^{E^{(3)}}\cdot U_{\text{max}}^{E^{(4)}}$ & $3.60744$\\
        $T^{\otimes5}\cdot U_{\text{max}}^{E^{(2)}}\cdot U_{\text{max}}^{E^{(3)}}\cdot U_{\text{max}}^{E^{(4)}}\cdot U_{\text{max}}^{E^{(5)}}$ & $3.59638$\\
        $U_{\text{max}}^{E^{(2)}}\cdot U_{\text{max}}^{E^{(3)}}$ & $3.23195$\\
        $U_{\text{max}}^{E^{(2)}}\cdot U_{\text{max}}^{E^{(3)}}\cdot U_{\text{max}}^{E^{(4)}}$ & $3.15329$\\
        $U_{\text{max}}^{E^{(2)}}\cdot U_{\text{max}}^{E^{(3)}}\cdot U_{\text{max}}^{E^{(4)}}\cdot U_{\text{max}}^{E^{(5)}}$ & $3.08596$\\
         \hline\hline
    \end{tabular}
    \caption{\textbf{Stabilizer extent of 5-qubit unitaries obtain by composing unitaries $U_{\max}^{E^{(k)}}$.} The stabilizer extent of each individual unitary can be found in the main text.}
    \label{tab: composition_chunks_of_k_uniform_candidates}
\end{table}

Another interesting point to consider is that the composition $T^{\otimes5}\cdot U_{\text{max}}^{E^{(2)}}$ is an instance of the so-called \textit{instantaneous quantum polynomial} (IQP) circuits capable of demonstrating quantum advantage. As shown in Ref.~\cite{bremner2016average}, efficient classical simulation (in the weak sense) is impossible for a class of IQP circuits composed only of $\{T,CS\}$ or $\{Z,CZ,CCZ\}$ (see also Ref.~\cite{Ni:2013tqo}). This is interesting, as the aforementioned composition is not the one with the maximum value of the stabilizer extent. In fact, of all the compositions presented in Table~\ref{tab: composition_chunks_of_k_uniform_candidates}, this is the one with the smallest stabilizer extent, making it the most easily simulable by the sum-over-Cliffords approach. These observations seem to connect to the work by Liu and Winter~\cite{LiuWinter2022manybody}, who pointed out that more magic does not necessarily equal more computational power. 

\section{Discussion and outlook}\label{sec: Conclusions}

In this paper, we exploit symmetry reductions to accelerate state-of-the-art classical algorithms for simulating Clifford-dominated circuits. Our approach yields more efficient methods for decomposing multiqubit, non-Clifford gates that are invariant under certain symmetry groups. This allowed us to find (previously out-of-reach) optimal expansions of several multiqubit quantum gates as sums of Clifford unitaries. These were then leveraged to demonstrate drastic improvements to the efficiency of sum-over-Clifford simulation for many cases of practical interest. In particular, we showed that our method leads to remarkable speed-ups in stabilizer-extent-based classical algorithms for simulating quantum Fourier transforms, instantaneous quantum polynomial circuits, and measurement-based quantum computations based on hypergraph states.

Beyond this, our method should be of further general interest in quantum computing problems where sequences of symmetric quantum gates appear naturally because of, e.g., operational constraints or existing problem structures. For instance, sequences of diagonal and real multiqubit gates have been proposed in applications such as quantum compiling and quantum approximate optimization~\cite{vandaele2025quantumbinaryfieldmultiplication,liuApplicationsCCZSGate2025}. Commuting (diagonal) gates are also native in quantum platforms such as ion traps and Rydberg atoms, and complex-sequences thereof are employed in quantum simulation of long-ranged quantum Ising models~\cite{Schauss_2018,PRXQuantum.4.010302,baezDynamicalStructureFactors2020,zhangObservationManybodyDynamical2017,islamOnsetQuantumPhase2011,bohnetQuantumSpinDynamics2016,bernienProbingManybodyDynamics2017}. Real and symmetric gates further appear commonly in fault tolerant quantum architectures, as transversal or fault-tolerant gates, e.g., while doing error correction with surface codes \cite{Fowler12SurfaceCodesPracticalQC,horsmanSurfaceCodeQuantum2012,Browne2007}.

Our work motivates further research into the classification of the mathematical symmetry groups that admit what we called \emph{strong symmetry reduction} (Theorem~\ref{theorem: Main result -- diagonals and reals}), which takes us to a new optimization search space of reduced size where further \emph{weak symmetry reductions} are feasible to implement (Lemma~\ref{lemma: weak reduction}). In particular, we leave the following open question: What is the underlying structure or physical property behind strong symmetry reduction? 
In quantum computation, group-theoretic methods, such as representation theory or group cohomology, have already been successful in the development of stabilizer-based classical simulation methods~\cite{BermejoVega_12_GKTheorem,bermejo-vegaNormalizerCircuitsGottesmanKnill2016,DBLP:journals/corr/Bermejo-Vega16}, the identification of computational phases of matter~\cite{elseSymmetryProtectedPhasesMeasurementBased2012,raussendorfComputationallyUniversalPhase2019,stephenSubsystemSymmetriesQuantum2019}, and in the study of quantum resources such as quantum contextuality~\cite{abramskySheaftheoreticStructureNonlocality2011,abramskyContextualityCohomologyParadox2015,raussendorfRoleCohomologyQuantum2023}. Developing further methods to understand symmetry in quantum circuits could yield new insights into the nature of quantum speed-ups, as well as better algorithms for classical simulation and hybrid quantum-classical computation~\cite{peresQuantumCircuitCompilation2023}.

A potential avenue for improving our methods would be studying the effect of approximation errors in the sum-over-Cliffords expansion. In situations where, even after leveraging strong symmetry reduction, the dimensionality of the convex problem remains unmanageable, we could try to approximate the expansion of the target gate using a smaller set of Clifford operators, trading computational cost for an error. Along a similar line, developing a heuristic algorithm with convergence guarantees reminiscent of the techniques proposed in~\cite{Hamaguchi2025} should also lead to relevant enhancements of our results. Lastly, as mentioned in Sec.~\ref{subsec: Statement of problem and main theoretical results}, weak symmetry reduction is applicable to any symmetry group with linear isometries. However, the cost of implementing this type of filtering has only been investigated in a reduced number of cases. Further improvements of our method can be achieved by designing better algorithms for performing the symmetry projection of a Clifford unitary.

A salient advantage of our symmetry reduction methods is their enhanced ability to find Clifford decompositions of sequences of quantum gates invariant under certain symmetry groups. This includes examples of important gates such as circuits of $T$, $CS$, and $CCZ$ gates. We observe that directly decomposing such sequences has a remarkably positive impact on the cost of classical simulation (e.g.~cf.~Fig.~\ref{fig: QFT submultiplicativity}). Underlying this advantage is the strict submultiplicativity of the stabilizer extent for these sequences -- a fundamental feature of this monotone, which has not been analyzed in detail earlier. This signals out the stabilizer extent as a potential quantifier to study the non-classical aspects of deep universal quantum circuits. Can the stabilizer extent capture quantum features that power quantum computation? How does this monotone relate to other notions of non-classicality that have been proposed for sequences of quantum transformations, such as, e.g., sequential contextuality~\cite{mansfieldQuantumAdvantageSequentialTransformation2018,emeriauQuantumAdvantageInformation2022}, transformation contextuality \cite{spekkensContextualityPreparationsTransformations2005,schmidContextualAdvantageState2018}, or the memory cost of quantum measurement sequences \cite{budroniMemoryCostTemporal2019,budroniContextualityMemoryCost2019}? An answer to these questions would lead to a deeper understanding of the foundations of quantum computational advantage.

\section*{Acknowledgements}

We thank David Gross, Hammam Qassim,  D. Luis Manuel Máñez Espina, Rhea Alexander-Turner, and José M. Martín for profitable research discussions. We acknowledge support from Ayuda Consolidación CNS2023-145392 (MICIU\slash AEI\slash 10.13039\slash 501100011033, NextGenerationEU\slash PRTR); Ramón y Cajal RYC2022-036209-I (MICIU\slash AEI\slash 10.13039\slash 501100011033, ESF+); EU HORIZON RIA FoQaCiA GA 101070558; Project Generation of Knowledge
PID 2024-162155OB-I00 (MICIU\slash AEI\slash 10.13039\slash 501100011033, ERDF\slash EU); Project FEDER C-EXP-256-UGR23 (Consejería de Universidad, Investigación e Innovación y UE Programa
FEDER Andalucía 2021-2027).
M.H.~acknowledges funding by the Deutsche Forschungsgemeinschaft (DFG, German Research Foundation) under Germany’s Excellence Strategy - Cluster of Excellence Matter and Light for Quantum Computing (ML4Q) EXC 2004/1 - 390534769.

\newpage
\clearpage
\appendix
\onecolumngrid

\section{Properties of the stabilizer extent for matrices}\label{app: Properties extent}

In this Appendix, we prove the properties of the stabilizer extent for matrices as stated in Proposition~\ref{prop: Properties extent} of the main text.

\begin{proof}[Proof of Proposition~\ref{prop: Properties extent}] We will prove the properties of the stabilizer extent for matrices one by one.
    
    (i) Let $A = \sum_{C\in \Cl{n}} x_C C$ be a unitary operator $A\in \mathrm{U}(2^n)$ so that $A A^{\dagger} = A^{\dagger} A = \one\,.$ This implies that
    \begin{equation*}
        \left| \tr \left[ A^{\dagger}A \right] \right| = 2^n \iff 2^n = \left| \tr \left[ \sum_{C,C^{\prime} \in \Cl{n}} x_{C^{\prime}} x_{C}^{*} C^{\dagger} C^{\prime} \right] \right| \,.
    \end{equation*}
    Analyzing the right-hand side, we note that:
    \begin{align*}
       \left| \tr \left[ \sum_{C,C^{\prime} \in \Cl{n}} x_{C^{\prime}} x_{C}^{*} C^{\dagger} C^{\prime} \right] \right| = \left|
       \sum_{C,C^{\prime} \in \Cl{n}} x_{C^{\prime}} x_{C}^{*} \tr \left[ C^{\dagger} C^{\prime} \right] \right|
         \leq 2^n \sum_{C,C^{\prime} \in \Cl{n}} \left| x_{C^{\prime}} \right| \left| x_{C}^{*}  \right|
         = 2^n \lVert \mathbf{x} \rVert_1^2\,,
    \end{align*}
    which implies that $\lVert \mathbf{x} \rVert_1^2 \geq 1\,.$ This means that any decomposition of a unitary operator must have an $\ell_1$-norm squared greater than or equal to one, thus, $\xi (A) \geq 1 \quad \forall A\in \mathrm{U}(2^n)\,.$ When $A$ is not only a unitary but a Clifford unitary, the trivial decomposition of $A$ (written as itself) achieves the minimum possible value of 1, thus $\xi(A) = 1 \quad \forall A\in \Cl{n}\,.$ From the arguments above, it is not hard to see that for any non-Clifford unitary $A$, $\xi (A) > 1$. The extent as defined in Def.~\ref{def: Stabilizer extent for unitaries} can only be smaller than 1 for non-unitary matrices.
    
    (ii) Let $A$ be a matrix with an optimal decomposition $A = \sum_{C\in \Cl{n}} x_C C$ meaning that $\xi (A) = \lVert \mathbf{x} \rVert_1^2\,.$ Suppose that $A$ is left-multiplied by a Clifford unitary $K\in \Cl{n}$ so that
    \begin{equation*}
        B = KA = \sum_{C\in \Cl{n}} x_C KC = \sum_{C^{\prime} \in \Cl{n}} x_{K^{\dagger}C^{\prime}} C^{\prime} \,.
    \end{equation*}
    By the definition of the stabilizer extent, the decomposition obtained above for $B$ must obey $\xi (B) \leq \lVert \mathbf{x} \rVert_1^2 = \xi (A)$\,. Next, suppose that the optimal decomposition of $B$ is $B = \sum_{C\in \Cl{n}} y_C C$ so that, by definition, $\xi (B) = \lVert \mathbf{y} \rVert_1^2\,.$ Since $B = K A$, then $A = K^{\dagger} B$ which implies that:
    \begin{equation*}
        A = \sum_{C\in \Cl{n}} y_C K^{\dagger} C = \sum_{C^{\prime}\in \Cl{n}} y_{KC^{\prime}} C^{\prime}\,.
    \end{equation*}
    Again, by the definition of the stabilizer extent, this decomposition implies that $\xi (A) \leq \lVert \mathbf{y} \rVert_1^2 = \xi (B)\,.$ Thus, we obtain that the stabilizer extent of $A$ and the stabilizer extent of $B=KA$ are the same. Similar calculations follow for right multiplication. Thus, the stabilizer extent for matrices is preserved under left and right multiplication by Clifford unitaries. 

    To prove monotonicity under stabilizer code projectors, consider an abelian subgroup of the Pauli group $\mathcal{G} \subset \mathcal{P}_n$ such that $\left|\mathcal{G} \right| = 2^r$. A projector onto the codespace defined by $\mathcal{G}$ has the form $\Pi_{\mathcal{G}} = 2^{-r} \sum_{P\in \mathcal{G}} P$. Then, we write
    \begin{equation*}
        \xi (\Pi_{\mathcal{G}} A) = \xi \left(2^{-r} \sum_{P\in \mathcal{G}} P A \right) 
        \leq 2^{-r} \sum_{P\in \mathcal{G}} \xi \left( P A \right)
        = 2^{-r} \sum_{P\in \mathcal{G}} \xi \left( A \right) = \xi (A)\,,
    \end{equation*}
    where we leveraged convexity (see below) and invariance under Cliffords. A similar reasoning allows to prove that $\xi ( A \Pi_{\mathcal{G}}) \leq \xi (A)$.

    (iii) To prove the submultiplicativity of the stabilizer extent under matrix multiplication consider two matrices $A,B\in \mathbb{C}^{2^n \times 2^n}$ with optimal decompositions $A = \sum_{C\in \Cl{n}} x_C C$ and $B = \sum_{C^{\prime}\in \Cl{n}} y_{C^{\prime}} C^{\prime}$ so that $\xi (A) = \lVert \mathbf{x} \rVert_1^2$ and $\xi (B) = \lVert \mathbf{y} \rVert_1^2$. Next, we take the product of these two matrices
    \begin{align*}
        AB = \sum_{C,C^{\prime} \in \Cl{n}} x_{C} y_{C^{\prime}} C C^{\prime} = \sum_{K \in \Cl{n}} \left( \sum_{C \in \Cl{n}} x_{C} y_{C^{\dagger}K} \right) K 
         = \sum_{K \in \Cl{n}} z_K K\,.
    \end{align*}
    By definition $\xi (AB) \leq \lVert \mathbf{z} \rVert_1^2\,$ and we can upperbound $\lVert \mathbf{z} \rVert_1$ as follows:
    \begin{align*}
        \lVert \mathbf{z} \rVert_1 = \sum_{K\in \Cl{n} } \left| \left( \sum_{C \in \Cl{n}} x_{C} y_{C^{\dagger}K} \right) \right| 
        \leq \sum_{K,C \in \Cl{n}} \left| x_{C} y_{C^{\dagger}K} \right| = \lVert \mathbf{x} \rVert_1 \lVert \mathbf{y} \rVert_1\,.
    \end{align*}
    Hence, $\xi (AB) \leq \xi(A)\xi(B)\,.$

    Submultiplicativity with respect to tensor products is fairly straightforward to show. As before, let us consider two matrices $A\in \mathbb{C}^{2^n \times 2^n}$, $B\in \mathbb{C}^{2^{m} \times 2^{m}}$ with optimal decompositions $A = \sum_{C\in \Cl{n}} x_C C$ and $B = \sum_{C^{\prime}\in \Cl{n^{\prime}}} y_{C^{\prime}} C^{\prime}$ so that $\xi (A) = \lVert \mathbf{x} \rVert_1^2$ and $\xi (B) = \lVert \mathbf{y} \rVert_1^2$. Considering their tensor product, we note that:
    \begin{equation*}
        A\otimes B = \sum_{C \in \Cl{n}} \sum_{C^{\prime} \in \Cl{n^{\prime}}} x_{C} y_{C^{\prime}} ( C \otimes C^{\prime})\,.
    \end{equation*}
    By definition $\xi (A\otimes B) \leq \left( \sum_{C \in \Cl{n}} \sum_{C^{\prime} \in \Cl{n^{\prime}}} \left| x_C y_{C^{\prime}} \right| \right)^2 = \xi(A)\xi(B)\,.$ 
    
    Showing that $\xi (A\otimes B) = \xi(B\otimes A)$ uses property (ii). Namely, there is always a suitable swapping operation, which we denote $\text{SWAP} \in \Cl{n+m}$, so that $\text{SWAP} (A\otimes B) \text{SWAP} = (B\otimes A)$. Hence, by property (ii), $\xi (A\otimes B) = \xi(\text{SWAP} (A\otimes B) \text{SWAP}) = \xi (B\otimes A)\,.$
    
    (iv) To demonstrate convexity of the stabilizer extent for matrices, consider again two matrices $A,B\in \mathbb{C}^{2^n \times 2^n}$ with optimal decompositions $A = \sum_{C\in \Cl{n}} x_C C$ and $B = \sum_{C^{\prime}\in \Cl{n}} y_{C^{\prime}} C^{\prime}$. Then, for any $t\in[0,1]$, $t A + (1-t)B = \sum_{C\in \Cl{n}} (t x_C + (1-t) y_C) C$ which implies that $\xi (tA+(1-t)B) \leq \lVert t \mathbf{x} + (1-t) \mathbf{y} \rVert_1^2 \leq t \xi(A) + (1-t) \xi(B)$ by the convexity of the norm and the square function.
    Homogeneity follows straightforwardly from Definition~\ref{def: Stabilizer extent for unitaries} and the homogeneity of the $\ell_1$-norm.

This concludes the proof of Proposition~\ref{prop: Properties extent}.
\end{proof}

We call the reader's attention to the proof of property~(i): If we were to restrict the definition of the stabilizer extent to unitaries, this would be a faithful measure in the sense that $\xi(A) = 1$ iff $A\in \Cl{n}$ and $\xi(A) >1$ otherwise. As it stands, Definition~\ref{def: Stabilizer extent for unitaries} applies to any arbitrary matrix and the stabilizer extent can therefore be smaller than 1; this is seen explicitly in Lemma~\ref{lemma: Characterization} and its proof in Appendix~\ref{app: Proof of characterization}.

\section{Lower bounds on gate synthesis}\label{app: Synthesis}

Here, we prove Proposition~\ref{prop: gate synthesis}, which explicitly states how the stabilizer extent can be used to lower bound the cost of exactly synthesizing certain unitaries.
\begin{proof}[Proof of Proposition~\ref{prop: gate synthesis}]
    By assumption, $s\in\mathbb{N}$ is the largest positive integer fulfilling the lower bound $\xi(T)^{s-1}<\xi(U)$, which implies  the upper bound $\xi(U)\leq\xi(T)^s$. Suppose wlog that the unitary gate of interest can be written as a quantum circuit of form $U = C_m L_m C_{m-1} L_{m-1} \cdots C_1 L_1 C_0$, where $\{C_j\}$ are Clifford unitaries and $\{L_j\}$ are non-Clifford layers containing $T$ gates acting on distinct qubits. Let $t$ be the total number of $T$ gates in this circuit. Then, the submultiplicativity of the extent implies $\xi(U)\leq \xi(T)^{t}$. Putting this together with our lower bound, it follows that $t>s-1$. This holds for any quantum circuit of Clifford+$T$ gates that can exactly synthesize $U$. Hence, $U$ cannot possibly be synthesized with  $s-1$ (or fewer) $T$ gates.  
  
    Let us consider that the optimal gate synthesis of $U$ in terms of the Clifford+$T$ gate set requires exactly $t_{\min}$ $T$ gates. Then, assigning these gates into the nonstabilizer layers and using Proposition~\ref{prop: Properties extent} leads to $\xi (U) = \xi (C_m L_m C_{m-1} L_{m-1} \cdots C_1 L_1 C_0) \leq \xi(T)^{t_{\min}}\,.$ and, necessarily, $t_{\min}\geq s$. Thus, the number of $T$ gates needed to synthesize $U$ is at least $s$.
\end{proof}

It is well-known that magic monotones can be used to lower bound the cost of (exact) synthesis of certain unitaries~\cite{Howard2017, Beverland2020lower}. For instance, a similar result was given in Ref.~\cite{Howard2017} using the robustness of magic, $\mathcal{R}$. In particular, let $\ket{U} = U\ket{+}^{\otimes n}$; if $\mathcal{R}(\ket{T}^{\otimes (s-1)}) < \mathcal{R}(\ket{U}) < \mathcal{R}(\ket{T}^{\otimes s})$, then a minimum number of $s$ $T$ gates are needed to (exactly) synthesize $U$. However, this result is limited by the submultiplicativity of the robustness of magic, which makes it hard to determine $\mathcal{R}(\ket{T}^{\otimes s})$ for values of $s$ beyond 10. Using the stabilizer extent has no such limitation, since $\xi (T^{\otimes s}) = \xi (T)^s$ (recall Sec.~\ref{subsec: Background - stabilizer extent}) and, therefore, the upper and lower bounds depend uniquely on $\xi (T) = \cos^{-2} (\pi/8) \approx 1.1716$.

The result in Proposition~\ref{prop: gate synthesis} can also be extended to other gate sets of the form  Clifford+$M$ provided that $M$ is a non-Clifford gate on at most three qubits which allows for a magic state $\ket{M}$ such that $\xi(M)=\xi(\ket{M})$ (recall discussion in Sec.~\ref{subsec: Background - stabilizer extent}).

\section{Proof of Lemma~\ref{lemma: Characterization}}\label{app: Proof of characterization}

Lemma~\ref{lemma: Characterization} is essential to the proof of our main result. 
In this appendix, we give a proof thereof, separated into three lemmas covering the three different subgroups.
The arguments use that stabilizer states have the following well-known form in the computational basis \cite{Gross2008nzg}:
\begin{equation}
\label{eq:stabilizer state in computational basis}
    \ket\psi = \frac{1}{\sqrt{|A|}}\sum_{\vec{x}\in A} i^{\vec{a}\cdot\vec{x}} (-1)^{q(\vec x) + \vec b \cdot \vec x} \ket{\vec x} \,,
\end{equation}
where $A\subset \Z_2^{n}$ is a suitable affine subspace, $\vec a\in \Z_2^{n}$, and $q$ is a quadratic function.
We note that the exponent of $i$ is meant to be evaluated modulo two.
Every stabilizer state has the form \eqref{eq:stabilizer state in computational basis}, and vice versa, any state of the form \eqref{eq:stabilizer state in computational basis} is a stabilizer state.

The matrix entries of a Clifford unitary can be deduced from Eq.~\eqref{eq:stabilizer state in computational basis} by referring to its Choi state $C\otimes\one \ket{\phi^+}$.
Since $\ket{\phi^+}:=2^{-n/2}\sum_{\vec{x}\in\Z_2^n}\ket{\vec x \vec x}$ is a stabilizer state, so is $C\otimes\one \ket{\phi^+}$, and we can write
\begin{equation*}
    C\otimes\one \ket{\phi^+}
    =
    \frac{1}{\sqrt{|W|}}\sum_{\vec{w}\in W} i^{\vec{a}\cdot\vec{w}} (-1)^{q(\vec w) + \vec b \cdot \vec w} \ket{\vec w} \,,
\end{equation*}
where $W\subset \Z_2^{2n}$ is a suitable affine subspace of dimension $n+s$ with $0\leq s\leq n$ \cite{Dehaene2003} ($s$ is the dimension of the support of every $C\ket{\vec x}$).
Writing $\vec w = (\vec x, \vec y)$, we then find the matrix entries of $C$ to be
\begin{equation}
\label{eq:Clifford matrix entries}
\bra{\vec x}C\ket{\vec y} 
= 
2^{n/2} \bra{\vec x \vec y}C\otimes\one\ket{\phi^+}
=
\frac{1}{2^{s/2}} i^{\vec{a}\cdot(\vec x, \vec y)} (-1)^{q(\vec x, \vec y) + \vec b \cdot (\vec x, \vec y)} \ind_W(\vec x, \vec y) \,.
\end{equation}
Here $\ind_W(w)$ is the indicator function on $W$ which is one if $w\in W$ and zero else.
Note that not every $2n$-qubit stabilizer state corresponds to a Clifford unitary; only the maximally entangled ones do.
By explicitly taking partial traces, it is straightforward to verify that this restricts the affine subspaces $W \in \Z_2^{2n}$ and the quadratic forms $q$, but not the linear part given by $\vec a$ and $\vec b$.
In particular, the matrix $C'$ obtained by setting $\vec a = 0$ is another valid Clifford unitary.

Finally, we note that the form of stabilizer states, and also Clifford unitaries, in the computational basis is only defined up to a global phase.
In the standard definition of the Clifford group as generated by Hadamard, CNOT, and phase gate, its center is $\Z_8$, and thus for any $C\in\Cl{n}$, we also have $\omega_8 C \in \Cl{n}$ for any eighth root of unity $\omega_8$.
However, for the stabilizer extent, this global phase does not matter: we can simply group the coefficients associated with the same Clifford unitary up to a phase.
Concretely, let $\overline{\Cl{n}}$ be the subset given by the phase convention \eqref{eq:Clifford matrix entries} and then write $U = \sum_{C\in\overline{\Cl{n}}} (\sum_{k=0}^7 x_C^{(k)} \omega_8^k) C $.
This procedure can only decrease the $\ell_1$-norm and thus any optimal decomposition can be chosen to only contain Cliffords from $\overline{\Cl{n}}$.

\begin{lemma}\label{lemma: characterization diagonal Cliffords}
    For all Clifford unitaries $C\in\Cl{n}$, we have $\xi_{\ZP{n}} \left( \Pi_{\ZP{n}} (C) \right)\leq1$.
\end{lemma}
Note that $\Pi_{\ZP{n}}$ is $\C$-linear, hence we do not need to explicitly consider a complex factor $x\in\C$.
\begin{proof} 
    For $C$ diagonal Clifford, the result is trivially true, as for $\Pi_{\ZP{n}} (C) = 0$.
    Hence, let us assume that $C$ is non-diagonal and that $d:=\Pi_{\ZP{n}} (C) \neq 0$.
    
    By Equation \eqref{eq:Clifford matrix entries}, the non-vanishing entries of $d$ then have the form
    \begin{equation}
    \label{eq:diagonal-coefficients}
        d_{\mathbf{xx}}=2^{-s/2} i^{\vec{a}'\cdot\vec{x}} (-1)^{q'(\vec{x})+\vec{b}'\cdot\vec{x}} \,,
    \end{equation}
    where $\vec{a}',\vec{b}'\in\Z_2^n$ and $q'$ are obtained from above in the obvious way.
    Next, we show that $d$ can be extended to a diagonal Clifford unitary.
    Let $S\subset\Z_2^n$ be the support of $\vec{x}\mapsto d_\vec{xx}$ and let us define another diagonal matrix $\kappa$ on the complement of $S$ by extending the linear and quadratic forms in Eq.~\eqref{eq:diagonal-coefficients} (in some way, this may not be unique).
    Then, the diagonal matrix $D=2^{s/2}(d+\kappa)$ is by construction Clifford because its diagonal is proportional to a (full-support) stabilizer state \cite{Dehaene2003,Gross2008nzg}.

    The support $S$ of $d$ is essentially the intersection of the affine support of $C\otimes\one\ket{\phi^+}$ with the diagonal $\Delta=\{(\vec x, \vec x) \; | \; \vec x\in\Z_2^n\}$, and thus also an affine subspace.
    We can hence write it as $S = V+\vec t$ for a linear subspace $V$ and a vector $\vec t$.
    $V^\perp$ is given as the subspace of vectors that are orthogonal to $V$ with respect to the binary dot product.
    Then, $G:=\{ (-1)^{\vec t\cdot \vec z} Z(\vec z) \; | \; \vec z \in V^\perp\}$ is an Abelian group which stabilizes $d$:
    \begin{equation*}
        (-1)^{\vec t\cdot \vec z} Z(\vec z)\cdot d
        =
        \sum_{\vec x\in S}\ (-1)^{\vec z \cdot (\vec x+\vec t)}\ d_{\vec{xx}} \ket{\vec x}\bra{\vec x}
        =
        d \,,
    \end{equation*}
    where the last step follows from $\vec x+\vec t \in V$ for all $\vec x \in S$.
    The average over $G$ annihilates $\kappa$:
    \begin{equation*}
      \frac{1}{|V^\perp|} \sum_{\vec z\in V^\perp}(-1)^{\vec t \cdot \vec z} Z(z)\cdot \kappa
      = \frac{1}{|V^\perp|}  \sum_{\vec x \in S^c} \underbrace{\sum_{\vec z\in V^\perp} (-1)^{\vec z \cdot (\vec x + \vec t)}}_{=|V^\perp|\delta(\vec x + \vec t\in V)}  \kappa_{\vec x \vec x} \ket{\vec x}\bra{\vec x} 
      = \sum_{\vec x\in S^c} \delta(\vec x\in S) \, \kappa_{\vec x \vec x} \ket{\vec x}\bra{\vec x} 
      = 0.
     \end{equation*}
     Hence, we get the desired decomposition of $d$ in terms of diagonal Cliffords:
     \begin{equation*}
      \frac{2^{-\frac{s}{2}}}{|V^\perp|} \sum_{\vec z \in V^\perp}(-1)^{\vec t \cdot \vec z } Z(\vec z )\cdot D 
      = \frac{1}{|V^\perp|} \sum_{\vec z  \in V^\perp}(-1)^{\vec{a} \cdot \vec z} Z(\vec z)\cdot (d+\kappa) 
      = d.
     \end{equation*}
     This shows that $\xi(d)=\xi(\Pi_{\ZP{n}} (C)) \leq 2^{-s} \leq 1$ as claimed.
     In fact, this is also optimal since we find for any decomposition $d = \sum_i c_i D^{(i)}$ into diagonal Cliffords $D^{(i)}$:
     \begin{equation*}
      2^{-s/2} = |d_{\vec x \vec x}| \leq \sum_i |c_i| \underbrace{|D^{(i)}_{\vec x \vec x}|}_{=1} = \|c\|_1.
     \end{equation*}
\end{proof}

\begin{lemma}\label{lemma: characterization real Cliffords}
    For all Clifford unitaries $C\in\Cl{n}$ and $x\in\C$, we have $\xi_{\K{n}} \left( \Pi_{\K{n}} (x C) \right)\leq |x|^2$.
\end{lemma}
\begin{proof}
We will show the following:
If the phase of $x$ is not a second or fourth root of unity, then $\Pi_{\K{n}} (xC) = 0$ and the statement is trivially true.
Otherwise, we construct a decomposition
\begin{equation}\label{eq:r_decomposition}
    \Pi_{\K{n}} (xC) = \alpha |x| \frac{1}{|G|}\sum_{P\in G}P C',
\end{equation}
where $0\leq\alpha\leq 1$, $C'$ is a real Clifford unitary, and the $P$'s are real Pauli operators from a suitable stabilizer group $G$.
From this, the claim $\xi_{\K{n}} (x C) \leq \alpha^2 |x|^2 \leq |x|^2$ follows directly.

For now, let us assume that $x=1$ and show $\xi_{\K{n}} \left( \Pi_{\K{n}} (C) \right)\leq 1$.
As we will see, it is pretty straightforward to modify our argument to cover the general case.
By Equation \eqref{eq:Clifford matrix entries}, the matrix entries of $C$ are of the form
\begin{equation*}
\bra{\vec x}C\ket{\vec y} 
=
\frac{1}{2^{s/2}} i^{\vec{a}\cdot(\vec x, \vec y)} (-1)^{q(\vec x, \vec y) + \vec b \cdot (\vec x, \vec y)} \ind_W(\vec x, \vec y) \,.
\end{equation*}
Recall that $\vec a, \vec b$, and $q$ are suitable linear and quadratic forms on a $(n+s)$-dimensional affine subspace $W\subset\Z_2^{2n}$.
The matrix $\eta = \Pi_{\K{n}} (C)$ then has the following entries:
\begin{equation*}
    \eta_{\vec x, \vec y}= 2^{-s/2} \ind_W(\vec x, \vec y)
    \begin{cases}
        (-1)^{q(\vec x, \vec y) + \vec b \cdot (\vec x, \vec y)} & \text{ if } \vec a \cdot (\vec x, \vec y)  = 0 \,,\\
        0 & \text{ otherwise}.
    \end{cases}
\end{equation*}
We define another real matrix $\kappa$ by
\begin{equation*}
    \kappa_{\vec x, \vec y} := 2^{-s/2} \ind_W(\vec x, \vec y)
    \begin{cases}
        (-1)^{q(\vec x, \vec y) + \vec b \cdot (\vec x, \vec y) } & \text{ if } \vec a \cdot (\vec x, \vec y)  \neq 0 \,,\\
        0 & \text{ otherwise}.
    \end{cases}
\end{equation*}
By construction, we have $C = \eta + i \kappa$. 
Moreover, if $C_{\vec y}$ denote the columns of $C$ and similarly $\eta_{\vec y}$ and $\kappa_{\vec y}$ those of $\eta$ and $\kappa$, respectively, then it is easy to check that
\begin{equation*}
    \frac{\one + (-1)^{\vec{v} \cdot \vec y} Z(\vec{u})}{2}\ket{C_{\vec y}} = \ket{\eta_{\vec y}} \quad \text{and} \quad
    \frac{\one - (-1)^{\vec{v} \cdot \vec y} Z(\vec{u})}{2}\ket{C_{\vec y}} = i\ket{\kappa_{\vec y}} \,.
\end{equation*}
Here, $\vec a = (\vec{u}, \vec{v})$.
This shows that the columns of $\eta$ and $\kappa$ are (potentially sub-normalized) stabilizer states, and we can deduce their stabilizer groups using standard arguments (cf.~Ref.~\cite[Sec.~4.4]{heinrich_stabiliser_2021}).
Note that the columns of $C$ form a \emph{stabilizer basis}, namely $\ket{C_{\vec y}}=C\ket{\vec y}$, thus their stabilizer groups $G_{\vec y}$ are the same up to signs.

\textbf{Case 1:}
Suppose $Z(\vec{u})$ commutes with all stabilizers, i.e.~it is contained in all $G_{\vec y}$ up to a sign that we can write as $(-1)^{\vec{m}\cdot\vec{y}+c}$.
Then, $C_{\vec y}$ is either left invariant by $\frac12(\one + (-1)^{\vec{v} \cdot \vec y+c} Z(\vec{u}))$ or is in its kernel, depending on whether $\vec{m}\cdot\vec{y} = \vec{v}\cdot\vec{y}$ or not.
\begin{enumerate}[label=(1\alph*)]
\item Assume that $\vec m = \vec v$. 
Then all columns of $C$ are left invariant and either $C=\eta$ or $C=i\kappa$, depending on $c=0$ or not. 
In the first case, $C$ is real and we can take $C=C'$, $\alpha=1$, $G=\{\one\}$.
In the second case, $C$ is purely imaginary, $\Pi_{\K{n}} (C)=0$, and we can take $\alpha=0$.
\item Assume that $\vec m \neq \vec v$ and for simplicity $c=0$.
Then, exactly half of the columns of $C$ are left invariant and become the columns of $\eta$ and the other half become the columns of $i\kappa$.
Note that the $\ket{\vec y}$ with $(\vec m + \vec v)\cdot \vec y = 0$ span a stabilizer code which is stabilized by $Z(\vec m + \vec v)$.
Thus, the same is true for the corresponding $\ket{C_{\vec y}}$ and the columns of $\eta$ are stabilized by $G:=\{\one,CZ(\vec m + \vec v)C^\dagger\}$.
Note that since $\eta$ is real, $G$ has to be a real stabilizer group.
By construction, the columns of $\kappa$ are $(-1)$-eigenvectors of $CZ(\vec m + \vec v)C^\dagger$ and thus 
\begin{equation*}
    \frac{1}{|G|} \sum_{h\in G} h (\eta+\kappa) = \eta = \Pi_{\K{n}} (C) \,.
\end{equation*}
Note that $\eta+\kappa = C'$ where $C'$ is a real matrix with entries 
\begin{equation*}
    C'_{\vec x, \vec y} = 2^{-s/2} (-1)^{q(\vec x, \vec y) + \vec b \cdot (\vec x, \vec y)} \ind_W(\vec x, \vec y) \,.
\end{equation*}
This is effectively the Clifford $C$ with $\vec a = 0$, thus $C'$ defines a valid Clifford unitary with real entries by the discussion after Eq.~\eqref{eq:Clifford matrix entries}.
Note that for $c=1$ the roles of $\eta$ and $\kappa$ are simply reversed, which does not change the result.
Setting $\alpha=1$, this concludes the construction of Eq.~\eqref{eq:r_decomposition} for case 1.
\end{enumerate}

\textbf{Case 2:}
Let us now assume that $Z(\vec{u})$ does not commute with all stabilizers.
Note that this implies that the norm of $\ket{\eta_{\vec y}}$ and $\ket{\kappa_{\vec y}}$ is $1/\sqrt{2}$, independent of other details.
Moreover, we can find Paulis $g_1,\dots,g_n$ (possibly with a $\pm$ sign) such that $Z(\vec{u})$ anti-commutes with $g_1$, it commutes with $g_2,\dots,g_n$, and $G_{\vec y} = \langle (-1)^{\vec{t}_1\cdot\vec{y}}g_1,\dots,(-1)^{\vec{t}_n\cdot\vec{y}}g_n\rangle$.
Here, $\vec{t}_i\in\Z_2^n$ is uniquely determined by $g_i$.
The stabilizer groups of the columns of $\eta$ and $\kappa$ are then given by $G_{\vec y}^\eta = \langle (-1)^{\vec{v}\cdot\vec y}Z(\vec{u}), (-1)^{\vec{t}_2\cdot\vec{y}}g_2,\dots,(-1)^{\vec{t}_n\cdot\vec{y}}g_n \rangle$ and $G_{\vec y}^\kappa  = \langle -(-1)^{\vec{v}\cdot\vec y}Z(\vec{u}), (-1)^{\vec{t}_2\cdot\vec{y}}g_2,\dots,(-1)^{\vec{t}_n\cdot\vec{y}}g_n \rangle$, respectively.
Note that since $\eta_{\vec y}$ and $\kappa_{\vec y}$ are proportional to real stabilizer states, their stabilizer groups are real, too.
We now have to further distinguish two subcases:
\begin{enumerate}[label=(2\alph*)]
    \item If $\vec v, \vec t_2, \dots, \vec t_n$ are linearly independent, then all columns of $\eta$ are mutually orthogonal (and so are those of $\kappa$).
    But then, $\Pi_{\K{n}} (C) = \eta = C'/\sqrt{2}$ for a real Clifford unitary $C'$ and Eq.~\eqref{eq:r_decomposition} holds with $\alpha=1/\sqrt{2}$ and $G=\{\one\}$.
    \item If $\vec v, \vec t_2, \dots, \vec t_n$ are linearly dependent, we can write $\vec v + \sum_{i=2}^{n} \lambda_i \vec t_i = 0$ for $\lambda_i\in\Z_2$.
    Thus 
    \begin{equation*}
        G_{\vec y} \ni
        (-1)^{\vec{v}\cdot\vec y}Z(\vec{u}) \prod_{i=2}^n \left[(-1)^{\vec{t}_2\cdot\vec{y}} g_2\right]^{\lambda_i} 
        =
        (-1)^{(\vec{v}+ \sum_{i=2}^{n} \lambda_i \vec t_i)\cdot\vec y} Z(\vec{u}) \prod_{i=2}^n g_2^{\lambda_i} 
        =
        Z(\vec{u}) \prod_{i=2}^n g_2^{\lambda_i} 
        =: g \,,
    \end{equation*}
    is an element contained in every stabilizer group $G_{\vec y}$.
    The real stabilizer group $G:=\langle g \rangle$ then simultaneously stabilizes all columns of $\eta$.
    Since every column of $\kappa$ is by construction a $(-1)$-eigenvector of $g$, the average over $G$ annihilates $\kappa$.
    Following the argument in case 1b, we can thus again choose the real Clifford $C'=\eta+\kappa$ and $\alpha=1$, which concludes the construction of Eq.~\eqref{eq:r_decomposition} for case 2.
\end{enumerate}

Finally, we argue for a general $x = |x| e^{i\varphi}$.
Since the matrix entries of $C$ have a phase that is a multiple of $\pi/4$, $\varphi$ has to be a multiple of $\pi/4$ itself, otherwise none of the matrix entries would be real and the projection would be identically zero, $\Pi_{\K{n}} (x C) = 0$, and the claim trivially holds.
Hence, we can assume that $x = |x| (-1)^k i^l$ where $k,l\in\Z_2$.
Since the projection $\Pi_{\K{n}}$ is $\R$-linear, we have $\Pi_{\K{n}}(|x| (-1)^k i^l C)=|x| (-1)^k \Pi_{\K{n}}(i^l C)$.
The multiplication by $i$ simply exchanges the role of $\eta$ and $\kappa$.
However, all of our arguments apply to both of them, implying a decomposition of the form
\begin{equation*}
    \Pi_{\K{n}} (i^l C) = \alpha \frac{1}{|G|}\sum_{P\in G}P C',
\end{equation*}
where the definition of $\alpha$, $G$, and $C'$ are as in above cases, possibly with $\eta$ and $\kappa$ exchanged.
The general case now follows by multiplying this equation with $|x| (-1)^k$ and noting that we can absorb a potential minus sign in the definition of $C'$.
\end{proof}

\begin{lemma}\label{lemma: characterization diag+real Cliffords}
    For all Clifford unitaries $C\in\Cl{n}$, we have $\xi_{\ZP{n}.\K{n}} \left( \Pi_{\ZP{n}.\K{n}} (xC) \right)\leq |x|^2$.
\end{lemma}
\begin{proof}
    Note that $\ZP{n}$ and $\K{n}$ commute, thus $\Pi_{\ZP{n}.\K{n}}=\Pi_{\K{n}}\Pi_{\ZP{n}}$ and we can modify the proof of Lemma \ref{lemma: characterization diagonal Cliffords} as follows.
    Taking the real part of $d=\Pi_{\ZP{n}} (C)$ results in real diagonal matrix $d_r$ with coefficients as in Eq.~\eqref{eq:diagonal-coefficients}, but with $\vec a'=0$ and the affine support $S_r=S\cap\{\vec x \, | \, \vec a'\cdot \vec x = 0 \}$.
    Now continue as in the proof of Lemma \ref{lemma: characterization diagonal Cliffords} and note that $D_r=2^{s/2}(d_r + \kappa_r)$ is a real diagonal matrix, so are the elements of $G$.
\end{proof}

\section{\texorpdfstring{Computing the cardinality of $\Di{n}/\Sn{n}$ and $\RD{n}/\Sn{n}$}{Computing the cardinality of Dn/Sn}}\label{app: cardinality of Dn/Sn}

The quotient set $\Di{n}/\Sn{n}$ denotes the set of all orbits from the action of $\Sn{n}$ on $\Di{n}$ by conjugation. Its cardinality can be computed through a direct use of Burnside's theorem, i.e., by averaging the number of fixed points in $\Di{n}$ by the action of each permutation in $\Sn{n}$. That is to say,
\begin{equation*}
    |\Di{n}/\Sn{n}|=\frac{1}{n!}\sum_{\sigma\in \Sn{n}}\ |\Di{n}^{\sigma}|,
\end{equation*}
where $\Di{n}^{\sigma}$ denotes the set of elements of $\Di{n}$ fixed by $\sigma$.

To calculate this number, we start with some considerations about $\Sn{n}$. It is well known that a permutation $\sigma$ can be uniquely decomposed into disjoint cyclic components, such that its effect can be understood by the way it cycles collections of qubits. A cycle of length $k$ will contain $k$ qubits, say with labels $\{i_1,i_2,...,i_k\}$, and $\sigma$ will act on it with the shifts $i_1\rightarrow i_2, i_2\rightarrow i_3,\dots, i_k\rightarrow i_1$, if the cycle appears in the decomposition of such a permutation. Each $\sigma$ will then correspond to a different partition of the $n$ qubit labels. However, as we will see, for the purposes of counting the number of elements in $\Di{n}^\sigma$, it only matters how many cycles $m_k$ of length $k$ appear on the decomposition of $\sigma$. Let $\lambda\vdash n$ (i.e. $\lambda$ partition of $n$) denote the collection of those permutations with the same number $m_k$ $\forall k$. The previous sum can then be rewritten as
\begin{equation*}
    |\Di{n}/\Sn{n}|=\frac{1}{n!}\sum_{\lambda\vdash n}\ \frac{n!}{\prod_k\ k^{m_k}m_k!}|\Di{n}^{\sigma\in\lambda}|,
\end{equation*}
where the factor $n!/\prod_kk^{m_k}m_k!$ counts the number of permutations with same partition $\lambda$. (Starting from $n!$ possible permutations, we divide by $m_k!\ \forall k$ to account for the permutation of similar cycles and by $k^{m_k}\ \forall k$ to account for the possible $k$ (cyclic) relabels within each cycle of length $k$.)

We now compute $|\Di{n}^{\sigma\in\lambda}|$. Notice that any diagonal Clifford unitary $D\in\Di{n}$ can be written as
\begin{equation}\label{eq:D_in_terms_of_S_CZ}
    D=\prod_{j=1}^n\ S^{a_j}_j\ \prod_{k>j}^n\ CZ^{b_{jk}}_{jk},
\end{equation}
where $a_j\in\{0,1,2,3\}$ and $b_{jk}\in\{0,1\}$ establishes the presence or absence of a $CZ$ gate between qubits $j$ and $k$. This highlights that $\Di{n}$ is isomorphic to $\mathbb{Z}_4^n\times\mathbb{Z}_2^{\binom{n}{2}}$. Now, for a permutation $\sigma\in\lambda$ to fix $D$, $a_j$ must be constant within each cycle of $\sigma$ and $b_{jk}$ must be constant within an orbit of pairs $\{j,k\}$ under $\sigma$. The first condition tells us that, if $\sigma$ has $l=\sum_k\ m_k$ cycles, then there are $4^l$ possible choices for the coefficients $a_j$. The second condition informs us that there is a factor of 2 for every orbit of pairs $\{j,k\}$. The number of such orbits within a given cycle of length $k$ is $k/2-1$ if $k$ is odd and $k/2$ if it is even. On the other hand, between cycles of lengths $k_1$ and $k_2$, the number of orbits is equal to the \textit{greatest common divisor} ($\mathrm{gcd}$) between them. Hence, denoting by $f(\sigma)$ the number of these orbits, we finally have that
\begin{equation*}
    |\Di{n}^{\sigma\in\lambda}|=4^{\sum_k m_k}\ 2^{f(\sigma)},
\end{equation*}
with
\begin{equation*}
    f(\sigma) =\sum_k\ m_k\ \bigg\lfloor \frac{k}{2}\bigg\rfloor+k\ \binom{m_k}{2} +\sum_{k_1<k_2}\ m_{k_1}\ m_{k_2}\ \mathrm{gcd}(k_1,k_2).
\end{equation*}
Finally, this allows us to write that
\begin{equation*}
    |\Di{n}/\Sn{n}|=\sum_{\lambda\vdash n}\ \frac{4^{\sum_km_k}\ 2^{f(\sigma\in\lambda)}}{\prod_k\ k^{m_k}m_k!}.
\end{equation*}

To conclude, we note that the same reasoning could be followed to find the cardinality of the quotient set $\RD{n}/\Sn{n}$, the difference being the reduction of possible phases to 2 instead of 4. Table~\ref{tab:cardinality_of_dn_sn} shows the cardinality of the quotient sets $\Di{n}/\Sn{n}$ and $\RD{n}/\Sn{n}$ for some values of $n$.
\begin{table}[ht]
    \centering
    \begin{tabular}{c c c}
        \hline\hline
        $n$ & $|\Di{n}/\Sn{n}|$ & $|\RD{n}/\Sn{n}|$\\
         \hline
        $1$ & $4$ & $2$\\
        $2$ & $20$ & $4$\\
        $3$ & $128$ & $16$\\
        $4$ & $1\,616$ & $117$\\
        $5$ & $32\,768$ & $1\,215$\\
        $6$ & $1\,516\,480$ & $19\,305$\\
        $7$ & $164\,003\,840$ & $438\,900$\\
        $8$ & $25\,434\,312\,192$ & $14\,373\,432$\\
         \hline\hline
    \end{tabular}
    \caption{\textbf{Cardinality of $\Di{n}/\Sn{n}$  and $\RD{n}/\Sn{n}$ for small values of $n$.}}
    \label{tab:cardinality_of_dn_sn}
\end{table}

\section{Additional reductions via permutation invariance}\label{app: Permutation invariance}

\begin{table*}[h]
    \centering
\begin{tabular}{c c c c c c c c c c c}
        \hline\hline
        $n$ & 0 $CZ$s & 1 $CZ$ & 2 $CZ$s & 3 $CZ$s & 4 $CZ$s & 5 $CZ$s & 6$CZ$s  & 7 $CZ$s & $\cdots$ & Total\\
         \hline
        $2$ & \underline{\textbf{16}} & 16\cellcolor{gray!25} & --- & --- & --- & --- & --- & --- & --- & 32\\
        $3$ & 64 & \underline{\textbf{192}} & 192\cellcolor{gray!25} & 64\cellcolor{gray!25} & --- & --- & --- & --- & --- & 512\\
        $4$ & $256$ & $1\,536$ & $3\,840$ & \underline{$\mathbf{5\,120}$} & \cellcolor{gray!25}$3\,840$ & \cellcolor{gray!25}$1\,536$ & \cellcolor{gray!25}256 & --- & --- & $16\,384$ \\
        $5$ & $1\,024$ & $10\,240$ & $46\,080$ & $122\,880$ & $215\,040$ & \underline{$\mathbf{258\,048}$} & \cellcolor{gray!25}$215\,040$ & \cellcolor{gray!25}$122\,880$ & \cellcolor{gray!25}$\cdots$ & $1\,048\,576$ \\
        $6$ & $4\,096$ & $61\,440$ & $430\,080$ & $1\,863\,680$ & $5\,591\,040$ & $12\,300\,288$ & $20\,500\,480$ & \underline{$\mathbf{26\,357\,760}$} & \cellcolor{gray!25}$\cdots$ & $134\,217\,728$ \\
         \hline\hline
    \end{tabular}
    \caption{\textbf{Number of diagonal Clifford operators for different numbers of qubits arranged into categories defined by the $CZ$ count}. 
    The numbers in each row are filled up to the existing category for that $n$. The underlined bold values refer to the category with the largest number of elements for each $n$. The shadowed values highlight the symmetric nature of the numbers in each now; specifically, counting the number of non-isomorphic, simple graphs with $k$ edges (for given $n$) is equivalent to counting non-isomorphic, simple graphs without $k$ edges from the maximum (for that $n$). The last column shows the total for each $n$ and corresponds to $\left|\Di{n}\right|$.}
    \label{tab: cardinality_of_dn_by_category}
\end{table*}

\begin{table*}[h]
    \centering
\begin{tabular}{c c c c c c c c c c c c}
        \hline\hline
        $n$ & 0 $CZ$s & 1 $CZ$ & 2 $CZ$s & 3 $CZ$s & 4 $CZ$s & 5 $CZ$s & 6$CZ$s  & 7 $CZ$s & $\cdots$ & Total & Fraction \\
         \hline
        $2$ & \underline{\textbf{16}} & 16\cellcolor{gray!25} & --- & --- & --- & --- & --- & --- & --- & 32 & 100\%\\
        $3$ & 64 & \underline{\textbf{64}} & 64\cellcolor{gray!25} & 64\cellcolor{gray!25} & --- & --- & --- & --- & --- & 256 & 50\%\\
        $4$ & $256$ & $256$ & $512$ & \underline{$\mathbf{768}$} & \cellcolor{gray!25}$512$ & \cellcolor{gray!25}$256$ & \cellcolor{gray!25}$256$ & --- & --- & $2\,816$ & $17.2\%$ \\
        $5$ & $1\,024$ & $1\,024$ & $2\,048$ & $4\,096$ & $6\,144$ & \underline{$\mathbf{6\,144}$} & \cellcolor{gray!25}$6\,144$ & \cellcolor{gray!25}$4\,096$ & \cellcolor{gray!25}$\cdots$ & $34\,816$ & $3.3\%$ \\
        $6$ & $4\,096$ & $4\,096$ & $8\,192$ & $20\,480$ & $36\,864$ & $61\,440$ & $86\,016$ & \underline{$\mathbf{98\,304}$} & \cellcolor{gray!25}$\cdots$ & $638\,976$ & $0.48\%$ \\
         \hline\hline
    \end{tabular}
    \caption{\textbf{Number of diagonal Clifford operators for different numbers of qubits arranged into categories defined by the $CZ$ count obtained by eliminating connectivity redundancies}. The notation follows that of Table \ref{tab: cardinality_of_dn_by_category}. The last column presents the fraction of the original set of diagonal Clifford unitaries that survives the first filtering stage.}
    \label{tab: cardinality_of_dn_with_graph_connectivity_moded_out}
\end{table*}

\begin{table*}[h]
    \centering
    \begin{tabular}{c c c c c c c c c c c c}
        \hline\hline
        $n$ & 0 $CZ$s & 1 $CZ$ & 2 $CZ$s & 3 $CZ$s & 4 $CZ$s & 5 $CZ$s & 6$CZ$s  & 7 $CZ$s & $\cdots$ & Total & Fraction \\
         \hline
        $2$ & \underline{\textbf{10}} & \cellcolor{gray!25}10 & --- & --- & --- & --- & --- & --- & --- & 20 & 62.5\%\\
        $3$ & 20 & \underline{\textbf{40}} & \cellcolor{gray!25}40 & \cellcolor{gray!25}20 & --- & --- & --- & --- & --- & 120 & 23.4\%\\
        $4$ & 35 & 100 & 215 & \underline{\textbf{296}} & \cellcolor{gray!25} 215 & \cellcolor{gray!25} 100 & \cellcolor{gray!25} 35 & --- & --- & 996 & 6.08\%\\
        $5$ & 56 & 200 & 620 & $1\,464$ & $2\,384$ & \underline{$\mathbf{2\,760}$} & \cellcolor{gray!25}$2\,384$ & \cellcolor{gray!25}$1\,464$ & \cellcolor{gray!25}$\cdots$ & $12\,208$ & 1.16\%\\
        $6$ & 84 & 350 & $1\,350$ & $4\,380$ & $11\,226$ & $22\,530$ & $35\,734$ & \underline{$\mathbf{45\,106}$} & \cellcolor{gray!25}$\cdots$ & $241\,520$ & 0.180\%\\
         \hline\hline
    \end{tabular}
    \caption{\textbf{Number of different $\Sn{n}$-projections in each category}. The notation follows that of Table \ref{tab: cardinality_of_dn_by_category}. The last column informs the fraction of the original set of diagonal Clifford operators that survives the complete filtering procedure described in the main text. We see that the reductions obtained are very impressive. If needed, these results can be improved further by considering equivalences across different categories.}
    \label{tab: cardinality_of_twirls_by_class}
\end{table*}

Appendix~\ref{app: cardinality of Dn/Sn} showed how to compute the number of orbits due to the action by the group of qubit permutations $\Sn{n}$ 
on the group of diagonal and real-diagonal Clifford operators. This gives us an upper bound for the number of required terms in decompositions leveraging Lemma~\ref{lemma: weak reduction}. However, in practice, several of these terms can be 
eliminated once we notice that diagonals that are equivalent by conjugation through some permutation yield the same $\Sn{n}$-projection.

Looking at Eq.~\eqref{eq:D_in_terms_of_S_CZ}, we observe that conjugation by a permutation preserves the number of $CZ$ gates. Thus, we can restrict ourselves to searching for nonequivalent unitaries within categories described by the number of $CZ$ gates that are present on them. Table~\ref{tab: cardinality_of_dn_by_category} shows the total number of diagonal Cliffords present in each such category before any elimination has taken place. From a permutation-invariance perspective, nonequivalent diagonal Cliffords within a $CZ$ category can be found by noticing that the problem boils down to finding non-isomorphic (under vertex-permutation) simple graphs with the same number of edges. The results are presented in Table~\ref{tab: cardinality_of_dn_with_graph_connectivity_moded_out}. To push things even further, an additional, internal elimination can be performed within the resulting elements accounted in Table~\ref{tab: cardinality_of_dn_with_graph_connectivity_moded_out}, when all possible phase configurations are taken into account. The final number of terms that have to be considered when leveraging Lemma~\ref{lemma: weak reduction} for the particular case of permutation-invariant diagonal unitaries are shown in Table~\ref{tab: cardinality_of_twirls_by_class}. A direct comparison with the values presented in Table~\ref{tab:cardinality_of_dn_sn} shows us that, in this case, the upper bound is rather weak.

For small $n$ (typically up to 10), the underlying graph problem is relatively easy and quick to solve numerically, and even ready-to-go datasets can be found on the internet. In our case, we have used the software \textit{SageMath}~\cite{sage} to yield the graphs. As a final comment, notice that we have only removed redundancies within each individual category. Because we can have equal projections across different categories, the values presented in Table~\ref{tab: cardinality_of_twirls_by_class} can still be further reduced, if necessary.

\section{Practical considerations on code implementation}\label{app: practical_considerations}

Memory is the main bottleneck associated with the optimization problem under consideration. This is because, as mentioned in the main text, the size of the search space scales (super)exponentially with $n$ (see Fig.~\ref{fig: Cardinality}). To give the reader an impression of the practical demands associated with the calculation of the stabilizer extent, in Table~\ref{tab: Space and time comparisons different problems} one can check the amount of memory required to \emph{formulate} the optimization problem for computing $\xi(C^{n-1}Z)$, using different sets of Clifford operators, namely, the entire Clifford group, the group of diagonal Clifford unitaries, and the group of real-diagonal Cliffords. This clearly highlights the beneficial impact of Theorem~\ref{theorem: Main result -- diagonals and reals} memory-wise. Additionally, we also understand that the time it takes to \emph{solve} the problem is not a concern.

Given this memory bottleneck, it is important to optimize its use. For the particular case of diagonal Clifford operators, by noticing that the non-zero entries of the matrices always have real and/or imaginary parts of $\{\pm 1\}$, we can use the most economic data type available, namely \textit{int8}, to store them. The same observation holds for real-diagonal Cliffords whose non-zero entries are always $\pm 1$.

An additional observation is that the memory needed by the optimizer to actually \emph{solve} the problem is larger than that needed to simply \emph{formulate} the problem. This is seen explicitly in Table~\ref{tab: Size of An}, where we present the use of memory when computing $\xi(T^{\otimes n})$ with $\Di{n}$ as search space. The values presented in this case are without leveraging our weak symmetry reduction results. They become much better when Lemma~\ref{lemma: weak reduction} is exploited (as can be understood by a direct comparison with Table~\ref{tab: Space and time comparisons different problems}).

We conclude by remarking that, when we want to use the weak symmetry reduction lemma, time may become the main bottleneck, not because of the optimization problem itself, but due to the intermediate step of filtering vectors with a given symmetry (as done for presenting the results in Appendix~\ref{app: Permutation invariance}). Fortunately, this kind of filtering procedure benefits immensely from parallelization. For our particular case, parallelization (and compilation) implemented through the use of \textit{Numba}~ \cite{numba} was sufficient, even though there is room to make our sub-routines faster.

\begin{table*}[h]
    \centering
    \begin{tabular}{c c c c c c c c}
        \hline\hline
        $n$ & 3 & 4 & 5 & 6 & 7 & 8 & 9  \\
         \hline
        Integer $A_n$ & 8.19 kB & 0.52 MB & 67.1 MB & 17.2 GB & 8.80 TB & 9.01 PB & 18.4 EB \\
        Problem definition & 443 kB & 25.0 MB & 3.07 GB & 475 GB* & --- & --- & ---\\
        Problem resolution & 2.0 MB & 12.5 MB & 1.18 GB & --- & --- & --- & ---\\
         \hline\hline
    \end{tabular}
    \caption{\textbf{Memory requirements when computing $\xi(T^{\otimes n})$ using only diagonal Clifford operators as a function of the number of qubits.} The first row presents the memory needed to store a matrix whose columns are the $\left|\Di{n}\right|$ diagonal Clifford unitaries encoded as vectors using the data type \textit{int8}. This matrix, denoted $A_n$, is used in the formulation of the optimization problem for diagonal unitaries. The values in this first row were calculated by hand. The second row indicates the memory needed by Gurobi to precisely formulate the entire optimization problem, as given by the size of the output \textit{.mps} file. Finally, the last row tells us the amount of extra memory needed while actually solving the problem, i.e., the incremented value of RAM due to resolution, which has been estimated by standard \textit{memory-profiling} for Python. Empty spaces indicate instances for which we were unable to obtain results. The asterisk indicates a value that was obtained using the Proteus supercomputer.}
    \label{tab: Size of An}
\end{table*}

\begin{table*}[h]
    \centering
    \begin{tabular}{c c c c c c c}
        \hline\hline
        $n$ & 2 & 3 & 4 & 5 & 6 & 7\\
         \hline
        Memory, $\Cl{n}$ & 11.5 MB & 88.6 GB* & --- & --- & --- & ---\\
        Memory, $\Di{n}$ & $12$ kB & $165$ kB & $2.394$ MB & $48.582$ MB & $2.69$ GB & ---\\
        Memory, $\RD{n}$ & $3$ kB & $12$ kB & $66$ kB & $686$ kB & $14.136$ MB & $572$ MB\\
        \hline
        Time, $\Cl{n}$ &  300 ms & --- & --- & --- & --- & ---\\
        Time, $\Di{n}$ &  $10$ ms & $20$ ms & $30$ ms & $320$ ms & $10.55$ s& ---\\
        Time, $\RD{n}$ &  $10$ ms & $10$ ms & $10$ ms & $30$ ms & $120$ ms & $4.32$s\\
         \hline\hline
    \end{tabular}
    \caption{\textbf{Comparison of the computational resources required to compute $\xi(C^{n-1}Z)$ when using different decomposition sets.} The presented time to solve is given by Gurobi, while the memory values are the size of the \textit{.mps} files underlying each optimization problem. The results for $\Di{n}$ and $\RD{n}$ employ both the strong and weak reductions, that is, Theorem \ref{theorem: Main result -- diagonals and reals} and Lemma \ref{lemma: weak reduction}, respectively. The value marked with * identifies an instance for which the problem was not solved because the optimizer was unable to allocate the required amount of RAM.}\label{tab: Space and time comparisons different problems}
\end{table*}

\section{nonequivalent graph states of five vertices}\label{app: graph states 5 vertices}

We considered all nonequivalent, simple graphs of 5 qubits and, for each, calculated the stabilizer extent of the unitary generated by associating with each edge the $CS$ gate. From this thorough analysis, ten different values of the stabilizer extent emerged. Fig.~\ref{fig: 2-uniform hypergraphs} summarizes the results, depicting the simplest graph instance associated with each of the ten values obtained.
\begin{figure*}[h]
    \centering
    \includegraphics[width=0.95\textwidth]{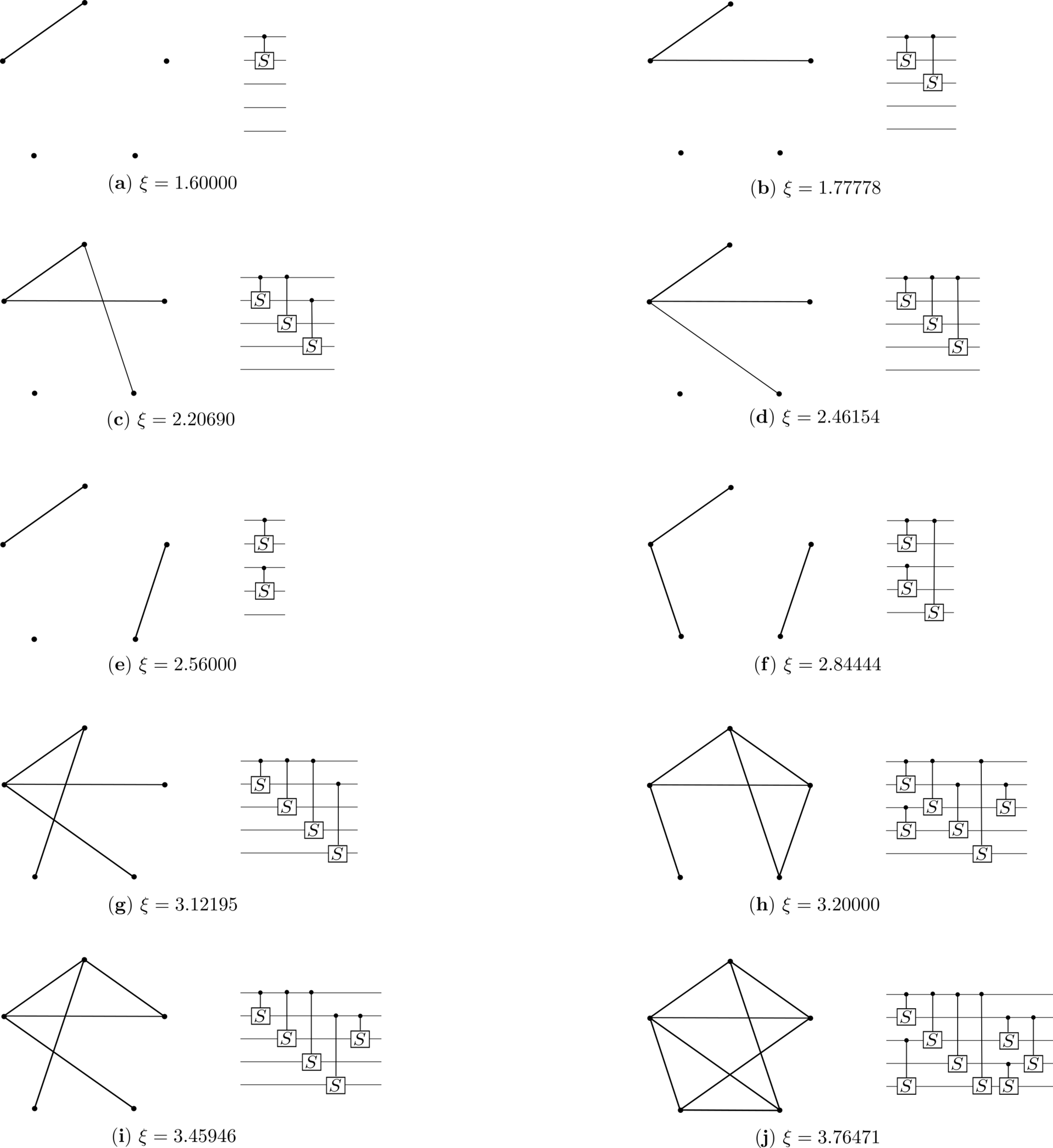}
    \caption{\textbf{Selection of simple graphs of five vertices and the stabilizer extent of the corresponding unitaries generated by associating a qubit to each vertex and applying a $CS$ gate to each pair of connected vertices.} Each different value of $\xi$ among all nonequivalent graphs defines what we call a family. Here, we depict a single representative for each family, choosing the graph with fewer edges. Note that by applying these unitaries to the state $\ket{+}^{\otimes 5}$, we obtain instances of 2-uniform generalized hypergraph states of five qubits.}
    \label{fig: 2-uniform hypergraphs}
\end{figure*}

\clearpage

\end{document}

%% file: paper_preamble.tex
\usepackage[
text={7.05in,10in},centering,
top = .84in, bottom=.84in
]{geometry}

\usepackage{graphicx} 
\usepackage[caption=false]{subfig} 
\usepackage{array} 
\usepackage{multirow}
\usepackage{comment}

\usepackage{bm}
\usepackage[tbtags]{amsmath}
\usepackage{amsthm}
\usepackage{nicefrac}
\usepackage{bbm}
\usepackage{physics}

\usepackage{tikz}
\usetikzlibrary{quantikz2} 

\usepackage{complexity} 
\usepackage{enumitem}

\usepackage{xcolor,colortbl}
\usepackage[colorlinks=true,linktocpage=true]{hyperref}%
\hypersetup{
    allcolors  = {NavyBlue},
}

\theoremstyle{plain}
\newtheorem{theorem}{Theorem}

\newtheorem{lemma}{Lemma}
\newtheorem{proposition}{Proposition}

\theoremstyle{plain}
\newtheorem{definition}{Definition}

\theoremstyle{definition}

\newcommand{\ugrshort}{Quantum Thermodynamics and Computation Group, Electromagnetism and Matter Physics Department, University of Granada, Granada, Spain}

\newcommand{\one}{\mathbbm{1}}
\newcommand{\id}{\mathrm{id}}

\newcommand{\Cl}[1]{\mathcal{C}_{#1}}
\newcommand{\Di}[1]{\mathcal{D}_{#1}}

\newcommand{\RD}[1]{\mathcal{RD}_{#1}}
\newcommand{\Rl}[1]{\mathcal{R}_{#1}}
\newcommand{\Sn}[1]{\mathcal{S}_{#1}}
\newcommand{\Sym}[1]{\text{Sym}_{#1}}
\renewcommand{\K}[1]{\mathcal{K}_{#1}}
\newcommand{\T}[1]{\mathcal{T}_{#1}}
\newcommand{\ZP}[1]{\mathcal{Z}_{#1}}
\newcommand{\Z}{\mathbb{Z}}
\renewcommand{\R}{\mathbb{R}}
\renewcommand{\C}{\mathbb{C}}
\renewcommand{\vec}[1]{\mathbf{#1}}
\renewcommand{\G}{\mathcal{G}}
\newcommand{\ind}{\mathsf{1}}
\newcommand{\Sp}{\mathrm{Sp}}

\usepackage{soul}